\documentclass[10pt]{article}

\usepackage{authblk}
\usepackage{natbib}
\usepackage{mathtools}
\usepackage{amsmath}
\usepackage{amsthm}
\usepackage{amssymb}
\usepackage{amsbsy}
\usepackage{amsfonts}
\usepackage{amscd}
\usepackage{mathrsfs}
\usepackage{bm}
\usepackage{breqn}
\usepackage{color, soul}
\usepackage{enumerate}
\usepackage{empheq}
\usepackage{rotating}
\usepackage{ftnxtra}
\usepackage{fnpos}
\usepackage[titletoc,toc,title]{appendix}
\usepackage{euscript}
\usepackage{graphicx}
\usepackage{epsfig}
\usepackage{epstopdf}
\DeclareGraphicsExtensions{.pdf,.png,.jpg,.eps}
\usepackage{pstool}
\usepackage{upgreek}
\usepackage{mathrsfs}
\usepackage{psfrag}

\numberwithin{equation}{section}

\theoremstyle{plain}	
\newtheorem{thm}{Theorem}[section]

\newtheorem{prop}[thm]{Proposition}
\newtheorem*{prop*}{Proposition} 
\theoremstyle{definition}	

\newtheorem{remark}[thm]{Remark}
\newtheorem{example}[thm]{Example}

\usepackage{graphicx}
\usepackage{lipsum}

\usepackage{caption}
\usepackage{subcaption}

\usepackage{hyperref}
\hypersetup{colorlinks=true, linkcolor=blue}
\hypersetup{colorlinks=true,citecolor=blue}

\DeclareMathAlphabet{\mathpzc}{OT1}{pzc}{m}{it}

\usepackage{amsmath, amsthm, amssymb}

\usepackage{cleveref}

\DeclarePairedDelimiter\abs{\lvert}{\rvert}
\makeatletter
\newsavebox{\@brx}
\newcommand{\llangle}[1][]{\savebox{\@brx}{\(\m@th{#1\langle}\)}%
  \mathopen{\copy\@brx\mkern2mu\kern-0.9\wd\@brx\usebox{\@brx}}}
\newcommand{\rrangle}[1][]{\savebox{\@brx}{\(\m@th{#1\rangle}\)}%
  \mathclose{\copy\@brx\mkern2mu\kern-0.9\wd\@brx\usebox{\@brx}}}%
\let\oldabs\abs
\def\abs{\@ifstar{\oldabs}{\oldabs*}}
\makeatother

\usepackage{accents}

\newcommand{\Fe}{\accentset{e}{\mathbf{F}}}

\newcommand{\Fa}{\accentset{a}{\mathbf{F}}}

\newcommand{\n}{\accentset{1}{\boldsymbol{\mathsf{n}}}}
\newcommand{\nn}{\accentset{2}{\boldsymbol{\mathsf{n}}}}
\newcommand{\nnn}{\accentset{3}{\boldsymbol{\mathsf{n}}}}

    {\end{bmatrix}}%

\usepackage{enumitem}

\usepackage[utf8]{inputenc}
\usepackage{dutchcal}
\usepackage{multirow}

\renewcommand{\arraystretch}{2.0}

\usepackage{xcolor}

\usepackage{setspace}

\usepackage{tikz}
\usetikzlibrary{positioning,arrows.meta,calc}

\begin{document}

\title{\textbf{
Universal deformations and universal residual stresses in incompressible isotropic Cauchy elasticity
}}

\author[1,2]{Arash Yavari\thanks{Corresponding author, e-mail: arash.yavari@ce.gatech.edu}}
\affil[1]{\small \textit{School of Civil and Environmental Engineering, Georgia Institute of Technology, Atlanta, GA 30332, USA}}
\affil[2]{\small \textit{The George W. Woodruff School of Mechanical Engineering, Georgia Institute of Technology, Atlanta, GA 30332, USA}}

\maketitle

\begin{abstract}
We study universal deformations in incompressible isotropic Cauchy elastic solids with residual stress, without assuming any specific origin for the residual stress. Starting from the constitutive representation of the Cauchy stress as an isotropic tensor-valued function of strain and residual stress, we derive the universality constraints for residually-stressed incompressible isotropic Cauchy elastic solids. We show that for the six known families of universal deformations the set of universal deformations is identical to that of incompressible isotropic elasticity in the absence of residual stress. We also show that residual stress does not enlarge the space of universal deformations. We then determine the universal residual stress fields corresponding to the six known families of universal deformations. Assuming that the residual stress field has the same symmetry as the corresponding universal deformation, the universality constraints reduce to systems of ordinary differential equations that can be solved explicitly. The resulting universal residual stress fields are characterized and discussed for each family. 
\end{abstract}

\begin{description}
\item[Keywords:] Universal deformation, Cauchy elasticity, hyperelasticity, Green elasticity, isotropic solids, residual stress.
\end{description}

\tableofcontents

%-----------------------------
%-----------------------------
\section{Introduction}

A \emph{universal deformation} is a deformation that can be sustained in the absence of body forces in every material belonging to a prescribed class by applying suitable boundary tractions. Although the boundary tractions required to maintain the deformation depend on the particular material, the deformation itself is independent of the constitutive model. The notion of universal deformations, also known as ``controllable" or ``general" deformations, originated in the pioneering works of Ericksen \citep{Ericksen1954,Ericksen1955}, which were motivated by the earlier studies of Rivlin \citep{Rivlin1948,Rivlin1949a,Rivlin1949b}. Ericksen showed that all universal deformations of homogeneous compressible isotropic solids are homogeneous \citep{Ericksen1955}. In the incompressible isotropic case, he identified four nontrivial families of universal deformations \citep{Ericksen1954}.
Ericksen further conjectured that any deformation with constant principal invariants must be homogeneous. This conjecture was shown to be false \citep{Fosdick1966}. Shortly thereafter, a fifth family of universal deformations, consisting of inhomogeneous deformations with constant principal invariants, was discovered independently by \citet{SinghPipkin1965} and \citet{KlingbeilShield1966}. Whether additional inhomogeneous universal deformations with constant principal invariants exist remains an open problem---Ericksen's Problem \citep{Marris1970,Kafadar1972,Marris1975,Marris1982,Fosdick1969,Fosdick1971}.

Universal deformations have played a fundamental role in nonlinear elasticity and related theories. They have been extensively used in semi-inverse methods for constructing exact solutions \citep{Knowles1979,Tadmor2012,Goriely2017}, in the design of experiments aimed at identifying constitutive equations \citep{Rivlin1951,Saccomandi2001}, and in the analysis of distributed eigenstrains and defects \citep{Wesolowski1968,Gairola1979,Zubov1997,YavariGoriely2012a,YavariGoriely2012b,YavariGoriely2013a,YavariGoriely2014,Golgoon2018,YavariGoriely2013b,YavariGoriely2015,Golgoon2017,Yavari2021Eshelby}. Universal deformations have also served as benchmark problems for computational methods in nonlinear elasticity \citep{Dragoni1996,Saccomandi2001,Chi2015,Shojaei2018} and have been employed in the determination of effective properties of nonlinear composite materials \citep{Hashin1985,Lopez2012,Golgoon2021}.

In linear elasticity, the analogous notion is that of a \emph{universal displacement} \citep{Truesdell1966,Gurtin1972,Carroll1973,Yavari2020}. The dependence of universal displacements on material symmetry has been studied systematically. Among linearly elastic solids, isotropic materials have the largest class of universal displacements, whereas triclinic materials possess the smallest. These investigations have subsequently been generalized to inhomogeneous solids \citep{YavariGoriely2022}, Cauchy elasticity \citep{YavariSfyris2025}, and linear anelasticity \citep{Yavari2022Anelastic-Universality}.

More recently, \citet{SfyrisYavari2026} extended the study of universal displacements to three-dimensional linear strain-gradient elasticity. For the Toupin--Mindlin first strain-gradient theory and the complete symmetry classification of strain-gradient elastic solids, they determined the universal displacement fields for all $48$ material symmetry classes. It was shown that for several high-symmetry classes, including isotropic solids, the strain-gradient universality constraints do not impose any additional restrictions beyond those of classical linear elasticity. For lower-symmetry classes, however, the higher-order elastic constants introduce additional universality constraints, and the corresponding universal displacement fields form proper subsets of the classical universal displacement families

Ericksen's theory has been extended in several directions, including anisotropic elasticity \citep{YavariGoriely2021,Yavari2021,YavariGoriely2023Universal}, implicit elasticity \citep{Yavari2024ImplicitElasticity}, Cauchy elasticity \citep{Yavari2024Cauchy}, and anelasticity \citep{YavariGoriely2016,Goodbrake2020}. Although Cauchy elasticity admits a broader constitutive structure than Green elasticity, the corresponding universal deformations and universal inhomogeneities are identical in the two theories. This equivalence was recently extended to anisotropic Cauchy elasticity by \citet{MotaghianYavari2026}, who studied universal deformations and universal material preferred directions in transversely isotropic, orthotropic, and monoclinic solids. They showed that the universality constraints of Cauchy elasticity and hyperelasticity are equivalent. Consequently, the universal deformations and the corresponding universal material preferred directions coincide with those of the corresponding anisotropic hyperelastic solids. These results demonstrate that, for these symmetry classes, universality is governed by material symmetry rather than by the existence of a strain-energy function. In the context of compressible anelasticity, universal deformations were shown to be covariantly homogeneous \citep{YavariGoriely2016}. For incompressible anelastic solids, the universal eigenstrain distributions associated with the six known families of universal deformations were determined by \citet{Goodbrake2020}. More recently, universality has been investigated in accreting bodies \citep{YavariPradhan2022,YavariAccretion2023,PradhanYavari2023} and in liquid crystal elastomers \citep{LeeBhattacharya2023,MihaiGoriely2023}.

The influence of internal constraints on universality has also been investigated, particularly in fiber-reinforced solids. \citet{Beskos1972} studied homogeneous compressible isotropic solids reinforced by inextensible fibers and examined whether the universal deformations of incompressible isotropic elasticity remain universal in this class of materials. An analogous investigation for incompressible isotropic solids was reported in \citep{Beskos1973}. For homogeneous compressible isotropic solids reinforced by a single family of inextensible fibers, \citet{Beatty1978,Beatty1989} determined the fiber distributions for which homogeneous deformations are universal. He showed that only three such distributions are possible, and that in each case the fibers are straight in both the reference and deformed configurations.

More recently, \citet{Yavari2023} analyzed universal displacements in compressible anisotropic linear elastic solids reinforced by a uniform distribution of inextensible straight fibers and characterized the corresponding classes of universal displacements for all compatible material symmetry classes. In a subsequent study, \citet{Yavari2025Fibers} presented the first systematic classification of universal deformations in compressible isotropic Cauchy elastic solids reinforced by a single family of inextensible fibers. For deformed fibers that remain straight, he identified a new inhomogeneous non-isochoric family of universal deformations---\emph{Family $Z_1$}. He further showed that if all principal invariants are constant, then only homogeneous universal deformations are possible. When the deformed fibers have non-vanishing curvature, the universality constraints become substantially more complicated, and the existence of universal deformations remains an open problem. Universality has also been studied in the presence of other internal constraints, including the inexpansibility constraint \citep{Kurashige1985} and in-plane rigidity \citep{Tommasi1996}.

In contrast to eigenstrains, the role of residual stress in universality has only recently been investigated. Residual stresses are self-equilibrated stresses that exist in the absence of external loads and arise in a wide range of natural and engineered solids through growth, plastic deformation, thermal processes, manufacturing, and other anelastic mechanisms \citep{Hoger1985,MerodioOgdenRodriguez2013,MerodioOgden2016}. Because residual stress enters the constitutive equations explicitly, it is natural to ask whether its presence enlarges the class of universal deformations.
For compressible isotropic Cauchy elastic solids with residual stress, this question was addressed by \citet{YavariMerodioShariff2025}. It was shown that any universal deformation must be homogeneous and that the associated universal residual stress field must also be homogeneous. Since a non-trivial residual stress field is necessarily inhomogeneous, it follows that universal residual stresses must vanish. Consequently, a compressible isotropic Cauchy elastic solid with a non-trivial residual stress distribution does not admit universal deformations. This result is consistent with the earlier work of \citet{YavariGoriely2016}, where universal eigenstrains were shown to be zero-stress (impotent).

The objective of the present paper is to study the incompressible counterpart of this problem. In contrast to the compressible case, incompressible isotropic solids admit at least five families of inhomogeneous universal deformations. We investigate whether these universal deformations remain universal in the presence of residual stress and determine the corresponding classes of universal residual stress fields. We also show that the presence of residual stress does not enlarge the space of universal deformations. 
For each of the six known families of universal deformations, we determine the most general universal residual stress fields that have the same symmetry as the corresponding universal deformations and are compatible with the universality constraints.

The present work is part of a broader research effort that we refer to as the \emph{universal program}. Figure \ref{Universal-Program} summarizes the notions of universality that have been studied to date and their extensions to different classes of materials and constitutive theories.

%-----------------------------
%-----------------------------
\begin{figure}[ht!]
\vskip 0.2in
\centering \includegraphics[width=.97\linewidth]{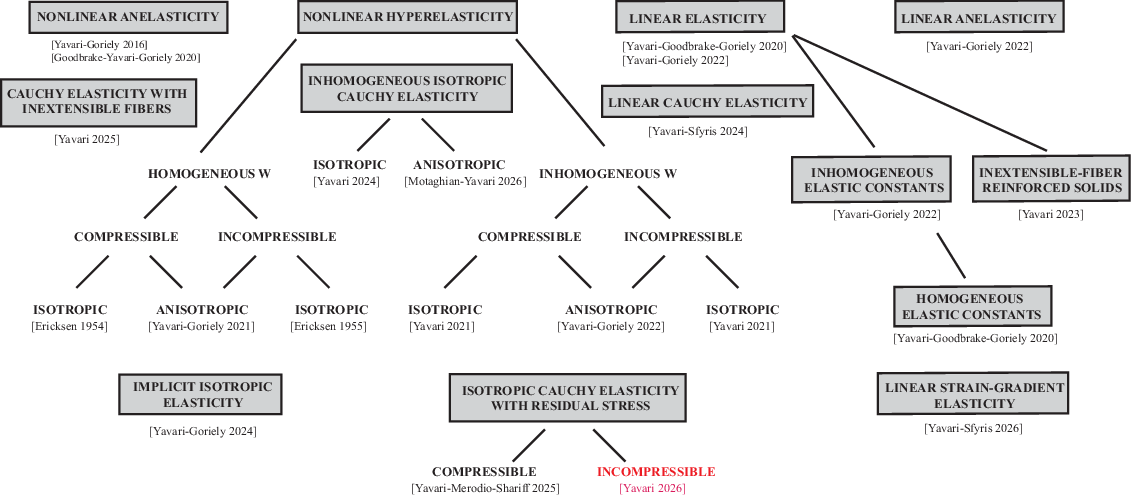}
\vskip 0.1in
\caption{
The universal program. The diagram summarizes the classification of universal deformations and related universal fields, including universal displacement fields, universal material preferred directions, universal inhomogeneities, universal eigenstrains, and universal residual stresses. Shown are the classes of materials and theories that have been studied to date, including nonlinear hyperelasticity, nonlinear Cauchy elasticity, nonlinear anelasticity, linear elasticity, linear Cauchy elasticity, linear strain-gradient elasticity, and linear anelasticity. The diagram also includes isotropic and anisotropic solids, compressible and incompressible materials, solids with homogeneous and inhomogeneous elastic constants, implicit elasticity, inextensible-fiber reinforced solids, and residually-stressed Cauchy elastic solids. 
The present work extends the universal program to incompressible isotropic Cauchy elastic solids with residual stress and determines the universal residual stress fields associated with the known families of universal deformations. The branch corresponding to the present work is highlighted in red.
}
\label{Universal-Program}
\end{figure}
%-----------------------------
%-----------------------------

This paper is organized as follows. In \S\ref{C-Elasticity}, we briefly review Cauchy elasticity and the constitutive representation of isotropic Cauchy elastic solids with residual stress. 
In \S\ref{C-Elasticity-UC}, we derive the universality constraints for incompressible isotropic Cauchy elastic solids with residual stress and show that residual stress does not enlarge the space of universal deformations. We also discuss the symmetry reduction of the resulting system of nonlinear universality PDEs and review the six known families of universal deformations of incompressible isotropic elasticity together with their symmetry structure.
In \S\ref{Universal-R-Stressed}, we determine the universal residual stress fields corresponding to the six known families of universal deformations. Assuming that the residual stress field has the same symmetry as the corresponding universal deformation, the universality constraints reduce to systems of ordinary differential equations that can be solved explicitly.
Conclusions are presented in \S\ref{Sec:Conclusions}.

%-----------------------------
%-----------------------------
\section{Cauchy Elasticity} \label{C-Elasticity}

Let us consider a body whose undeformed configuration is identified with an embedded submanifold $\mathcal{B}$ of the Euclidean ambient space $\mathcal{S}$. The flat metric of the ambient space is denoted by $\mathbf{g}$, and the induced metric on the reference configuration is $\mathbf{G} = \mathbf{g}\big|_{\mathcal{B}}$. A deformation is a smooth map $\varphi:\mathcal{B} \to \mathcal{C} \subset \mathcal{S}$, where $\mathcal{C} = \varphi(\mathcal{B})$ denotes the deformed configuration. The tangent map of $\varphi$ defines the deformation gradient $\mathbf{F} = T\varphi$, a metric-independent linear map $\mathbf{F}(X): T_X\mathcal{B} \to T_{\varphi(X)}\mathcal{C}$ at each material point $X \in \mathcal{B}$. 

In local coordinates $\{X^A\}$ on $\mathcal{B}$ and $\{x^a\}$ on $\mathcal{C}$, the deformation gradient has components $F^a{}_A = \frac{\partial \varphi^a}{\partial X^A}$. Its transpose $\mathbf{F}^{\mathsf{T}}$ has components $(F^{\mathsf{T}})^A{}_a = g_{ab}\, F^b{}_B\, G^{AB}$. The right Cauchy-Green strain tensor is defined by $\mathbf{C} = \mathbf{F}^{\mathsf{T}}\mathbf{F}$, with components $C^A{}_B = (F^{\mathsf{T}})^A{}_a\, F^a{}_B$. This gives the standard expression
%-----------------------------
\begin{equation}
	C_{AB} = (g_{ab} \circ \varphi)\, F^a{}_A\, F^b{}_B\,,
\end{equation}
%-----------------------------
which shows that the right Cauchy-Green strain is the pull-back of the spatial metric, i.e., $\mathbf{C}^\flat = \varphi^* \mathbf{g}$, where the flat operator is defined via the reference metric $\mathbf{G}$. 

The left Cauchy-Green strain tensor is defined as $\mathbf{B}^{\sharp} = \varphi^*(\mathbf{g}^{\sharp})$, with components $B^{AB} = F^{-A}{}_a\, F^{-B}{}_b\, g^{ab}$, where $F^{-A}{}_a$ are the components of $\mathbf{F}^{-1}$. The spatial counterpart of $\mathbf{C}^\flat$ is $\mathbf{c}^\flat = \varphi_*\mathbf{G}$, whose components are $c_{ab} = F^{-A}{}_a\, F^{-B}{}_b\, G_{AB}$. Likewise, the spatial analogue of $\mathbf{B}^{\sharp}$ is $\mathbf{b}^{\sharp} = \varphi_*(\mathbf{G}^{\sharp})$, with components $b^{ab} = F^a{}_A\, F^b{}_B\, G^{AB}$. Note that $\mathbf{b} = \mathbf{c}^{-1}$.

The tensors $\mathbf{C}$ and $\mathbf{b}$ share the same principal invariants $I_1$, $I_2$, and $I_3$, which are defined as follows \citep{Ogden1984,MarsdenHughes1994}:
%-----------------------------
\begin{equation}
\begin{aligned}
	I_1 &= \operatorname{tr}\mathbf{b} = b^{ab} g_{ab}\,, \\
	I_2 &= \frac{1}{2}\left(I_1^2 - \operatorname{tr} \mathbf{b}^2\right)
	= \frac{1}{2}\left(I_1^2 - b^{ab} \,b^{cd} g_{ac} \,g_{bd}\right)\,, \\
	I_3 &= \det \mathbf{b}\,.
\end{aligned}
\end{equation}
%-----------------------------

In Cauchy elasticity, the stress at a material point and at a given instant depends explicitly on the strain at that same point and time \citep{Cauchy1828,Truesdell1952,TruesdellNoll2004,YavariGoriely2025}. However, the existence of an energy function is not guaranteed.\footnote{It is important to emphasize that Cauchy elasticity does not describe all elastic materials. A more general class of elastic solids are described by implicit constitutive relations of the form $\boldsymbol{\mathsf{f}}(\boldsymbol{\sigma},\mathbf{b})=\mathbf{0}$ \citep{Morgan1966,Rajagopal2003,Rajagopal2007}. Cauchy elasticity is as a special case within this broader class.} In terms of the first Piola-Kirchhoff stress tensor \citep{Truesdell1952,TruesdellNoll2004,Ogden1984}, one has
%----------------------------------
\begin{equation} \label{P-Cauchy}
	\mathbf{P}=\hat{\mathbf{P}}(X,\mathbf{F},\mathbf{G},\mathbf{g})\,.
\end{equation}
%----------------------------------
Objectivity implies that the second Piola-Kirchhoff stress must take the form \citep{TruesdellNoll2004}
%----------------------------------
\begin{equation} \label{S-Cauchy}
	\mathbf{S}=\hat{\mathbf{S}}(X,\mathbf{C}^\flat,\mathbf{G})\,.
\end{equation}
%----------------------------------
For isotropic materials, one arrives at the classical representation \citep{RivlinEricksen1955,Wang1969,Boehler1977}
%---------------------------------
\begin{equation}
	\mathbf{S}=\Lambda_0 \mathbf{G}^\sharp+\Lambda_1 \mathbf{C}^\sharp
	+\Lambda_{-1} \mathbf{C}^{-\sharp}\,,
\end{equation}
%---------------------------------
where $\Lambda_i=\Lambda_i(X,I_1,I_2,I_3)$, $i=-1,0,1$, and $\sharp$ denotes the sharp operator associated with the metric $\mathbf{G}$. 

The Cauchy stress for compressible isotropic Cauchy elastic materials admits the representation
%---------------------------------
\begin{equation}
	\boldsymbol{\sigma} =\alpha\, \mathbf{g}^\sharp+\beta \mathbf{b}^\sharp
	+\gamma\mathbf{c}^\sharp\,, 
\end{equation}
%---------------------------------
where $\alpha=\alpha(X,I_1,I_2,I_3)$, $\beta=\beta(X,I_1,I_2,I_3)$, and $\gamma=\gamma(X,I_1,I_2,I_3)$ are arbitrary constitutive functions.

%-----------------------------
%----------------------------
\subsection{Cauchy elastic solids with distributed eigenstrains} 

An elastic body with eigenstrains is residually-stressed, in general. Deformation gradient is multiplicatively decomposed as $\mathbf{F}=\Fe\Fa$, where $\Fe$ and $\Fa$ are the elastic and anelastic distortions, respectively.\footnote{See \citep{SadikYavari2017, yavariSozio2023} for more details on this decomposition and its history.} These distortions are not compatible, in general.
Imagine that the body in its reference configuration is partitioned into a large number of small (infinitesimal) pieces. The anelastic distortion $\Fa$ maps each small piece to its local relaxed state.
Let us consider an infinitesimal line element in the reference configuration. It is represented as $\mathbf{M} dS$, where $\mathbf{M}$ is a unit vector. $\Fa\mathbf{M} dS$ is relaxed line element.
Suppose the Euclidean metric on the body is $\mathring{\mathbf{G}}$. The squared length of the relaxed line element is calculated as
%---------------------------------
\begin{equation}
	\llangle \Fa\mathbf{M} dS,\Fa\mathbf{M} dS\rrangle_{\mathring{\mathbf{G}}}
	=\llangle \mathbf{M} dS, \mathbf{M} dS\rrangle_{\Fa^*\mathring{\mathbf{G}}}\,,
\end{equation}
%---------------------------------
where $\mathbf{G}=\Fa^*\mathring{\mathbf{G}}=\Fa^\star\mathring{\mathbf{G}}\Fa$ is called the material metric \citep{YavariGoriely2013,Yavari2021Eshelby}. It is non-flat, i.e., it has non-vanishing Riemann curvature, in general.

In the presence of eigenstrains, the constitutive equations of a Cauchy anelastic solid has the form \eqref{P-Cauchy} or \eqref{S-Cauchy}, where $\mathbf{G}$ is the material metric, see \citep{YavariGoriely2025} for more details.

%-----------------------------
%----------------------------
\subsection{Cauchy elastic solids with residual stress} 

Consider a body made of a compressible Cauchy elastic solid that is isotropic in the absence of residual stress. 
We assume that the body is residually stressed and denote the second Piola--Kirchhoff stress by $\mathring{\mathbf{S}}$. The residual stress field corresponds to the embedding $\iota:\mathcal{B} \hookrightarrow \mathcal{S}$ whose right Cauchy--Green strain is $\mathring{\mathbf{C}}^\flat=\iota^*\mathbf{g}=\mathbf{G}$.
In the absence of body forces and boundary tractions, residual stress is self-equilibrated; hence $\operatorname{Div}_{\mathring{\mathbf{C}}^\flat}\mathring{\mathbf{S}}=\operatorname{Div}_{\mathbf{G}}\mathring{\mathbf{S}}=\mathbf{0}$ in $\mathcal{B}$ and $\mathring{\mathbf{S}}\mathbf{N}=\mathbf{0}$ on $\partial\mathcal{B}$, where $\mathbf{N}$ is the unit normal to the undeformed boundary.\footnote{Recall that  when the equilibrium equations are written in terms of the second Piola--Kirchhoff stress, the divergence operator is taken with respect to the metric $\mathbf{C}^\flat=\varphi^*\mathbf{g}$.} 
We denote the push-forward of $\mathring{\mathbf{S}}$ to the current configuration by $\mathring{\boldsymbol{\sigma}}$, i.e., $\mathring{\boldsymbol{\sigma}}=\varphi_*\mathring{\mathbf{S}}=\mathbf{F}\mathring{\mathbf{S}}\mathbf{F}^\star$.\footnote{Note that symmetry of $\mathring{\mathbf{S}}$ implies the symmetry of $\mathring{\boldsymbol{\sigma}}$, i.e., $\mathring{\boldsymbol{\sigma}}^\star=\mathring{\boldsymbol{\sigma}}$.}

The Cauchy stress $\boldsymbol{\sigma}$ is an isotropic tensor-valued function of the two symmetric second-order tensors $\mathbf{b}$ and $\mathring{\boldsymbol{\sigma}}$. As is well known from the theory of invariants such a function can be written in terms of a basis of nine independent scalar invariants and a corresponding representation involving symmetric tensor products and contractions \citep{Spencer1970,Smith1971,Boehler1979,Boehler1987}. Specifically, one has
%-----------------------------
\begin{equation} \label{Cauchy-Stress-Representation}
\begin{aligned}
	\boldsymbol{\sigma} &= \alpha\, \mathbf{g}^\sharp
	+ \beta\, \mathbf{b}^\sharp
	+ \gamma\, \mathbf{c}^\sharp
	+ \zeta\, \mathring{\boldsymbol{\sigma}}
	+ \omega\, \mathring{\boldsymbol{\sigma}}^2
	+ \xi\, (\mathbf{b} \, \mathring{\boldsymbol{\sigma}} + \mathring{\boldsymbol{\sigma}} \, \mathbf{b})  
	+ \eta\, (\mathbf{c}\, \mathring{\boldsymbol{\sigma}} + \mathring{\boldsymbol{\sigma}} \, \mathbf{c})
	+\chi \big(
	\mathbf{b}\,\mathring{\boldsymbol{\sigma}}^2+\mathring{\boldsymbol{\sigma}}^2\,\mathbf{b}
	\big)
	\,,
\end{aligned}
\end{equation}
%-----------------------------
where the scalar-valued response functions $\alpha, \beta, \gamma, \zeta, \omega, \xi, \eta, \chi$ depend smoothly on the following ten functionally independent invariants defined as\footnote{For a symmetric tensor $\mathbf{k}$, we have $k^a{}_b = k_b{}^a$, and hence both are denoted by $k^a_b$.}
%-----------------------------
\begin{equation}
\begin{aligned}
	I_1 &= \operatorname{tr} \mathbf{b}^\sharp = b^{ab}\, g_{ab}\,, \\
	I_2 &= \frac{1}{2} \left( I_1^2 - b^{ab} b^{cd}\, g_{ac}\, g_{bd} \right)\,, \\
	I_3 &= \det \mathbf{b}\,, \\
	I_4 &= \operatorname{tr} \mathring{\boldsymbol{\sigma}} = \mathring{\sigma}^{ab} g_{ab}\,, \\
	I_5 &= \operatorname{tr}(\mathring{\boldsymbol{\sigma}}^2) 
	= \mathring{\sigma}^{ac}\, \mathring{\sigma}_c^b\, g_{ab}\,, \\
	I_6 &= \operatorname{tr}(\mathring{\boldsymbol{\sigma}}^3) 
	= \mathring{\sigma}^{ad}\, \mathring{\sigma}_d^c\, \mathring{\sigma}_c^b\, g_{ab}\,, \\
	I_7 &= \operatorname{tr}(\mathbf{b} \, \mathring{\boldsymbol{\sigma}}) 
	= b^a_c\, \mathring{\sigma}^{cb}\, g_{ab}\,, \\
	I_8 &= \operatorname{tr}(\mathbf{b}^2 \, \mathring{\boldsymbol{\sigma}}) 
	= b^{ad}\, b_{dc}\, \mathring{\sigma}^{cb}\, g_{ab}\,, \\
	I_9 &= \operatorname{tr}(\mathbf{b} \, \mathring{\boldsymbol{\sigma}}^2) 
	= b^a_c\, \mathring{\sigma}^{cd}\, \mathring{\sigma}_d^b\, g_{ab}\,, \\
	I_{10} &= \operatorname{tr}(\mathbf{b}^2 \, \mathring{\boldsymbol{\sigma}}^2) 
	= b^a_c\,b^c_d\,\mathring{\sigma}^d_e\,\mathring{\sigma}^e_a \,.
\end{aligned}
\end{equation}
%-----------------------------
In coordinates, the Cauchy stress has the representation
%-----------------------------
\begin{equation}
\begin{aligned}
	\sigma^{ab} &= \alpha\, g^{ab}
	+ \beta\, b^{ab}
	+ \gamma\, c^{ab}
	+ \zeta\, \mathring{\sigma}^{ab}+ \omega\, \mathring{\sigma}^{an}\, \mathring{\sigma}_n^b
	+ \xi\, \left(b^a_n\, \mathring{\sigma}^{nb} + \mathring{\sigma}^{an}\, b_n^b \right)  
	+ \eta\, \left(c^a_n\, \mathring{\sigma}^{nb} + \mathring{\sigma}^{an}\, c_n^b \right)\\
	&\quad
	+\chi \left(b^a_n\, \mathring{\sigma}^{2nb} + \mathring{\sigma}^{2an}\, b_n^b \right)
	\,,
\end{aligned}
\end{equation}
%-----------------------------
where $\mathring{\sigma}^{2nb} =\mathring{\sigma}^n_k \,\mathring{\sigma}^{kb}$.

%-----------------------------
\begin{remark}
Although the material is assumed to be isotropic in the absence of residual stress, the presence of a nontrivial residual stress field $\mathring{\boldsymbol{\sigma}}$ induces an effective anisotropy. In particular, $\mathring{\boldsymbol{\sigma}}$ plays the role of a structural tensor, and the material symmetry group is reduced to the subgroup of $\mathrm{Orth}(\mathbf{G})$ that leaves $\mathring{\boldsymbol{\sigma}}$ invariant. Consequently, the constitutive response is no longer isotropic in the classical sense, but is instead an isotropic function of the pair $(\mathbf{b}, \mathring{\boldsymbol{\sigma}})$. The resulting symmetry class is determined by the spectral structure of $\mathring{\boldsymbol{\sigma}}$: a hydrostatic residual stress preserves isotropy, while more general residual stresses lead to transversely isotropic, orthotropic, or lower symmetry classes.
\end{remark}
%-----------------------------

For an incompressible Cauchy elastic solid with residual stress, $I_3=1$, one has the following representation for the Cauchy stress
%-----------------------------
\begin{equation} 
\begin{aligned}
	\boldsymbol{\sigma} &= (-p+\alpha)\, \mathbf{g}^\sharp
	+ \beta\, \mathbf{b}^\sharp
	+ \gamma\, \mathbf{c}^\sharp
	+ \zeta\, \mathring{\boldsymbol{\sigma}}
	+ \omega\, \mathring{\boldsymbol{\sigma}}^2
	+ \xi\, (\mathbf{b} \, \mathring{\boldsymbol{\sigma}} + \mathring{\boldsymbol{\sigma}} \, \mathbf{b})  
	+ \eta\, (\mathbf{c}\, \mathring{\boldsymbol{\sigma}} + \mathring{\boldsymbol{\sigma}} \, \mathbf{c})
	+\chi \big(\mathbf{b}\,\mathring{\boldsymbol{\sigma}}^2+\mathring{\boldsymbol{\sigma}}^2\,\mathbf{b}
	\big)
	\,,
\end{aligned}
\end{equation}
%-----------------------------
where $p$ is the Lagrange multiplier corresponding to the incompressibility constraint. 
As $p$ is an unknown at this stage, one can replace $-p+\alpha$ by $-p$, and hence
%-----------------------------
\begin{equation} \label{Cauchy-Stress-Representation-Incompressible}
\begin{aligned}
	\boldsymbol{\sigma} &= -p\, \mathbf{g}^\sharp
	+ \beta\, \mathbf{b}^\sharp
	+ \gamma\, \mathbf{c}^\sharp
	+ \zeta\, \mathring{\boldsymbol{\sigma}}
	+ \omega\, \mathring{\boldsymbol{\sigma}}^2
	+ \xi\, (\mathbf{b} \, \mathring{\boldsymbol{\sigma}} + \mathring{\boldsymbol{\sigma}} \, \mathbf{b})  
	+ \eta\, (\mathbf{c}\, \mathring{\boldsymbol{\sigma}} + \mathring{\boldsymbol{\sigma}} \, \mathbf{c})
	+\chi \big(\mathbf{b}\,\mathring{\boldsymbol{\sigma}}^2+\mathring{\boldsymbol{\sigma}}^2\,\mathbf{b}
	\big)
	\,,
\end{aligned}
\end{equation}
%-----------------------------
In components, we have
%-----------------------------
\begin{equation}
\begin{aligned}
	\sigma^{ab} &= -p\, g^{ab}
	+ \beta\, b^{ab}
	+ \gamma\, c^{ab}
	+ \zeta\, \mathring{\sigma}^{ab}+ \omega\, \mathring{\sigma}^{an}\, \mathring{\sigma}_n^b
	+ \xi\, \left(b^a_n\, \mathring{\sigma}^{nb} + \mathring{\sigma}^{an}\, b_n^b \right) \\
	&\quad + \eta\, \left(c^a_n\, \mathring{\sigma}^{nb} + \mathring{\sigma}^{an}\, c_n^b \right)
	+ \chi \left(b^a_n\, \mathring{\sigma}^n_k \,\mathring{\sigma}^{kb} 
	+ \mathring{\sigma}^{ak}\,\mathring{\sigma}^n_k\, b_n^b \right)
	\,.
\end{aligned}
\end{equation}
%-----------------------------

%-----------------------------
%-----------------------------
\section{Universal Deformations and Residual Stresses in Incompressible Isotropic Cauchy Elasticity} \label{C-Elasticity-UC}

For a residually-stressed Cauchy elastic solid, equilibrium equations in the absence of body forces read $\operatorname{div}\boldsymbol{\sigma}=\mathbf{0}$.
Substituting \eqref{Cauchy-Stress-Representation} into the equilibrium equations, one obtains
%-----------------------------
\begin{equation} \label{Equilibrium-RSCE}
\begin{aligned}
	\operatorname{div}\boldsymbol{\sigma} &= - \operatorname{grad} p
	+ \beta\, \operatorname{div}\mathbf{b}^\sharp
	+ \gamma\, \operatorname{div}\mathbf{c}^\sharp
	+ \mathbf{b}^\sharp\cdot \nabla\beta
	+ \mathbf{c}^\sharp\cdot \nabla\gamma
	+\zeta \operatorname{div}\mathring{\boldsymbol{\sigma}}
	+ \mathring{\boldsymbol{\sigma}}\cdot \nabla\zeta
	+ \omega\, \operatorname{div}\mathring{\boldsymbol{\sigma}}^2
	+ \mathring{\boldsymbol{\sigma}}^2\cdot\nabla\omega \\
	&\quad
	+ \xi\, \operatorname{div}
	(\mathbf{b} \cdot \mathring{\boldsymbol{\sigma}} + \mathring{\boldsymbol{\sigma}} \cdot \mathbf{b})  
	+ \eta\, \operatorname{div}
	(\mathbf{c}\cdot \mathring{\boldsymbol{\sigma}} + \mathring{\boldsymbol{\sigma}} \cdot \mathbf{c})
	+ (\mathbf{b} \cdot \mathring{\boldsymbol{\sigma}} + \mathring{\boldsymbol{\sigma}} \cdot \mathbf{b})\cdot\nabla\xi  
	+ (\mathbf{c}\cdot \mathring{\boldsymbol{\sigma}} + \mathring{\boldsymbol{\sigma}} \cdot \mathbf{c})\cdot\nabla\eta \\
	&\quad
	+ \big(\mathbf{b}\,\mathring{\boldsymbol{\sigma}}^2+\mathring{\boldsymbol{\sigma}}^2\,\mathbf{b}\big)\cdot \nabla\chi
	+\chi \operatorname{div}
	\big(\mathbf{b}\,\mathring{\boldsymbol{\sigma}}^2+\mathring{\boldsymbol{\sigma}}^2\,\mathbf{b}\big)
	=\mathbf{0}\,.
\end{aligned}
\end{equation}
%-----------------------------
In components, $\sigma^{am}{}_{|m}=0$, which can be rewritten as
%-----------------------------
\begin{equation}
\begin{aligned}
p_{,a} &= 
 \beta_{,m}\, b_a^m + \beta\, b_a^m{}_{|m}
+ \gamma_{,m}\, c_a^m + \gamma\, c_a^m{}_{|m}
+ \zeta\, \mathring{\sigma}_a^m{}_{|m}
+ \zeta_{,m}\, \mathring{\sigma}_a^m  
+ \omega_{,m}\, (\mathring{\sigma}_a^n\, \mathring{\sigma}_n^m)
+ \omega\Big(
      \mathring{\sigma}_a^n{}_{|m}\, \mathring{\sigma}_n^m
    + \mathring{\sigma}_a^n\, \mathring{\sigma}_n^m{}_{|m}
   \Big) \\
&\quad
+ \xi_{,m}\, \big(b_a^n\, \mathring{\sigma}_n^m + \mathring{\sigma}_a^n\, b_n^m\big)
+ \xi\Big(
      b_a^n{}_{|m}\, \mathring{\sigma}_n^m 
    + b_a^n\, \mathring{\sigma}_n^m{}_{|m}
    + \mathring{\sigma}_a^n{}_{|m}\, b_n^m
    + \mathring{\sigma}_a^n\, b_n^m{}_{|m}
   \Big) \\
&\quad
+ \eta_{,m}\, \big(c_a^n\, \mathring{\sigma}_n^m + \mathring{\sigma}_a^n\, c_n^m\big)
+ \eta\Big(
      c_a^n{}_{|m}\, \mathring{\sigma}_n^m 
    + c_a^n\, \mathring{\sigma}_n^m{}_{|m}
    + \mathring{\sigma}_a^n{}_{|m}\, c_n^m
    + \mathring{\sigma}_a^n\, c_n^m{}_{|m}
   \Big) \\
&\quad
+ \kappa_{,m}\, (b_a^n\, \mathring{\sigma}_n^d\, b_d^m)
+ \kappa\Big(
      b_a^n{}_{|m}\, \mathring{\sigma}_n^d\, b_d^m
    + b_a^n\, \mathring{\sigma}_n^d{}_{|m}\, b_d^m
    + b_a^n\, \mathring{\sigma}_n^d\, b_d^m{}_{|m}
   \Big) \\
&\quad
+ \chi_{|mb}\,\Big(b_a^n\,\mathring{\sigma}^{2m}_n+\mathring{\sigma}^{2n}_a\,b_n^m\Big) 
+ \chi_{,m}\,\Big(b_{a|b}^n\,\mathring{\sigma}^{2m}_n
+b_a^n\,\mathring{\sigma}^{2m}_{~n|b}
+\mathring{\sigma}^{2n}_{~a|b}\,b_n^m
+\mathring{\sigma}^{2n}_a\,b_n^m{}_{|b}\Big) \\
&\quad
+ \chi_{,b}\,\Big(b_{a|m}^n\,\mathring{\sigma}^{2m}_n
+b_a^n\,\mathring{\sigma}^{2m}_{~n|m}
+\mathring{\sigma}^{2n}_{~a|m}\,b_n^m
+\mathring{\sigma}^{2n}_a\,b_n^m{}_{|m}\Big) \\
&\quad
+ \chi\,\Big(b_{a|mb}^n\,\mathring{\sigma}^{2m}_n
+b_{a|m}^n\,\mathring{\sigma}^{2m}_{~n|b}
+b_{a|b}^n\,\mathring{\sigma}^{2m}_{~n|m}
+b_a^n\,\mathring{\sigma}^{2m}_{~n|mb} 
+ \mathring{\sigma}^{2n}_{~a|mb}\,b_n^m
+ \mathring{\sigma}^{2n}_{~a|m}\,b_n^m{}_{|b}
+ \mathring{\sigma}^{2n}_{~a|b}\,b_n^m{}_{|m}
+ \mathring{\sigma}^{2n}_a\,b_n^m{}_{|mb}\Big)\,.
\end{aligned}
\end{equation}
%-----------------------------

The integrability conditions for the existence of $p$ are $p_{,ab} = p_{,ba}$, equivalently $p_{|ab} = p_{|ba}$, which give the universality constraints.
Expanding each term in $p_{|ab}=(p_{,a})_{|b}$, one obtains
%-----------------------------
\begin{equation}
\begin{aligned}
p_{|ab} &= 
   \beta_{|mb}\, b_a^m 
 + \beta_{,m}\, b_a^m{}_{|b}
 + \beta_{,b}\, b_a^m{}_{|m} 
 + \beta\, b_a^m{}_{|mb} \\[4pt]
&\quad
 + \gamma_{|mb}\, c_a^m
 + \gamma_{,m}\, c_a^m{}_{|b}
 + \gamma_{,b}\, c_a^m{}_{|m}
 + \gamma\, c_a^m{}_{|mb} \\[4pt]
&\quad
 + \zeta_{|mb}\, \mathring{\sigma}_a^m
 + \zeta_{,m}\, \mathring{\sigma}_a^m{}_{|b}
 + \zeta_{,b}\, \mathring{\sigma}_a^m{}_{|m}
 + \zeta\, \mathring{\sigma}_a^m{}_{|mb} \\
 &\quad
 + \omega_{|mb} \left(\mathring{\sigma}_a^n\, \mathring{\sigma}_n^m \right) \\
&\quad
 + \omega_{,m} \left(\mathring{\sigma}_a^n{}_{|b}\, \mathring{\sigma}_n^m
                   + \mathring{\sigma}_a^n\, \mathring{\sigma}_n^m{}_{|b}\right) \\
&\quad
 + \omega_{,b} \left(\mathring{\sigma}_a^n{}_{|m}\, \mathring{\sigma}_n^m
                   + \mathring{\sigma}_a^n\, \mathring{\sigma}_n^m{}_{|m} \right) \\
&\quad
 + \omega \left(\mathring{\sigma}_a^n{}_{|mb}\, \mathring{\sigma}_n^m
              + \mathring{\sigma}_a^n{}_{|m}\, \mathring{\sigma}_n^m{}_{|b}
              + \mathring{\sigma}_a^n{}_{|b}\, \mathring{\sigma}_n^m{}_{|m}
              + \mathring{\sigma}_a^n\, \mathring{\sigma}_n^m{}_{|mb} \right)\,,\\
&\quad
 + \xi_{|mb}\,\left(b_a^n\, \mathring{\sigma}_n^m + \mathring{\sigma}_a^n\, b_n^m\right) \\
&\quad
 + \xi_{,m}\,\left(b_a^n{}_{|b}\, \mathring{\sigma}_n^m 
                + b_a^n\, \mathring{\sigma}_n^m{}_{|b} 
                + \mathring{\sigma}_a^n{}_{|b}\, b_n^m 
                + \mathring{\sigma}_a^n\, b_n^m{}_{|b}\right) \\
&\quad
 + \xi_{,b}\,\left(b_a^n{}_{|m}\, \mathring{\sigma}_n^m 
                 + b_a^n\, \mathring{\sigma}_n^m{}_{|m} 
                 + \mathring{\sigma}_a^n{}_{|m}\, b_n^m 
                 + \mathring{\sigma}_a^n\, b_n^m{}_{|m}\right) \\
&\quad
 + \xi\,\Big(b_a^n{}_{|mb}\, \mathring{\sigma}_n^m
           + b_a^n{}_{|m}\, \mathring{\sigma}_n^m{}_{|b}
           + b_a^n{}_{|b}\, \mathring{\sigma}_n^m{}_{|m}
           + b_a^n\, \mathring{\sigma}_n^m{}_{|mb} \\
&\qquad\quad
           + \mathring{\sigma}_a^n{}_{|mb}\, b_n^m
           + \mathring{\sigma}_a^n{}_{|m}\, b_n^m{}_{|b}
           + \mathring{\sigma}_a^n{}_{|b}\, b_n^m{}_{|m}
           + \mathring{\sigma}_a^n\, b_n^m{}_{|mb}\Big) \\
&\quad
 + \eta_{|mb}\,\left(c_a^n\, \mathring{\sigma}_n^m + \mathring{\sigma}_a^n\, c_n^m\right) \\
&\quad
 + \eta_{,m}\,\left(c_a^n{}_{|b}\, \mathring{\sigma}_n^m 
                 + c_a^n\, \mathring{\sigma}_n^m{}_{|b} 
                 + \mathring{\sigma}_a^n{}_{|b}\, c_n^m 
                 + \mathring{\sigma}_a^n\, c_n^m{}_{|b}\right) \\
&\quad
 + \eta_{,b}\,\left(c_a^n{}_{|m}\, \mathring{\sigma}_n^m 
                 + c_a^n\, \mathring{\sigma}_n^m{}_{|m} 
                 + \mathring{\sigma}_a^n{}_{|m}\, c_n^m 
                 + \mathring{\sigma}_a^n\, c_n^m{}_{|m}\right) \\
&\quad
 + \eta\,\Big(c_a^n{}_{|mb}\, \mathring{\sigma}_n^m
            + c_a^n{}_{|m}\, \mathring{\sigma}_n^m{}_{|b}
            + c_a^n{}_{|b}\, \mathring{\sigma}_n^m{}_{|m}
            + c_a^n\, \mathring{\sigma}_n^m{}_{|mb} \\
&\qquad\quad
            + \mathring{\sigma}_a^n{}_{|mb}\, c_n^m
            + \mathring{\sigma}_a^n{}_{|m}\, c_n^m{}_{|b}
            + \mathring{\sigma}_a^n{}_{|b}\, c_n^m{}_{|m}
            + \mathring{\sigma}_a^n\, c_n^m{}_{|mb}\Big) \\
&\quad
+ \chi_{|mb}\,\Big(b_a^n\,\mathring{\sigma}^{2m}_n+\mathring{\sigma}^{2n}_a\,b_n^m\Big) \\
&\quad
+ \chi_{,m}\,\Big(b_{a|b}^n\,\mathring{\sigma}^{2m}_n
+b_a^n\,\mathring{\sigma}^{2m}_{~n|b}
+\mathring{\sigma}^{2n}_{~a|b}\,b_n^m
+\mathring{\sigma}^{2n}_a\,b_n^m{}_{|b}\Big) \\
&\quad
+ \chi_{,b}\,\Big(b_{a|m}^n\,\mathring{\sigma}^{2m}_n
+b_a^n\,\mathring{\sigma}^{2m}_{~n|m}
+\mathring{\sigma}^{2n}_{~a|m}\,b_n^m
+\mathring{\sigma}^{2n}_a\,b_n^m{}_{|m}\Big) \\
&\quad
+ \chi\,\Big(b_{a|mb}^n\,\mathring{\sigma}^{2m}_n
+b_{a|m}^n\,\mathring{\sigma}^{2m}_{~n|b}
+b_{a|b}^n\,\mathring{\sigma}^{2m}_{~n|m}
+b_a^n\,\mathring{\sigma}^{2m}_{~n|mb} \\
&\qquad\qquad\ \ 
+ \mathring{\sigma}^{2n}_{~a|mb}\,b_n^m
+ \mathring{\sigma}^{2n}_{~a|m}\,b_n^m{}_{|b}
+ \mathring{\sigma}^{2n}_{~a|b}\,b_n^m{}_{|m}
+ \mathring{\sigma}^{2n}_a\,b_n^m{}_{|mb}\Big)\,.
\end{aligned}
\end{equation}
%-----------------------------
Knowing that $\beta=\beta(I_1,I_2,I_4,I_5,I_6,I_7,I_8,I_9,I_{10})$, with the shorthand notation $\beta_j = \partial \beta / \partial I_j$, we have
%-----------------------------
\begin{equation}
	\beta_{,m} =
	  \beta_{1}\, I_{1,m}+ \beta_{2}\, I_{2,m}	+ \beta_{4}\, I_{4,m}+ \cdots+ \beta_{10}\, I_{10,m}
	=\sum_j \beta_{j}\, I_{j,m}\,,
\end{equation}
%-----------------------------
where $j\in\{1,2,4,5,6,7,8,9,10\}$.
Thus, the second covariant derivative is written as
%-----------------------------
\begin{equation}
\beta_{|mb} 
= \sum_{j} \beta_j\, I_{j|mb}
+ \sum_{j,k} \beta_{jk}\, I_{j,m}\, I_{k,b}\,,
\end{equation}
%-----------------------------
where $\beta_j = \partial\beta/\partial I_j$, $\beta_{jk} = \partial^2\beta / (\partial I_j \partial I_k)$, and  $j,k\in\{1,2,4,5,6,7,8,9,10\}$.

The contributions of the seven response functions are written as
%-----------------------------
\begin{equation}
\begin{aligned}
	p_{|ab}^{(\beta)} =\beta\, b_a^m{}_{|mb} 
	+ \sum_j \beta_j  \Big( I_{j|mb}\, b_a^m + I_{j,m}\, b_a^m{}_{|b} + I_{j,b}\, b_a^m{}_{|m} \Big) 
	+ \sum_{j,k} \beta_{jk}\,   I_{j,m}\, I_{k,b}\, b_a^m \,,
\end{aligned}
\end{equation}
%-----------------------------
%-----------------------------
\begin{equation}
\begin{aligned}
	p_{|ab}^{(\gamma)} &=
	\gamma\, c_a^m{}_{|mb} 
	+ \sum_j \gamma_j  \Big( I_{j|mb}\, c_a^m + I_{j,m}\, c_a^m{}_{|b} + I_{j,b}\, c_a^m{}_{|m} \Big) 
	+ \sum_{j,k} \gamma_{jk}\,  I_{j,m}\, I_{k,b}\, c_a^m \,,
\end{aligned}
\end{equation}
%-----------------------------
%-----------------------------
\begin{equation}
	p_{|ab}^{(\zeta)} =\zeta\, \mathring{\sigma}_a^m{}_{|mb} 
	+ \sum_j \zeta_j  \Big( I_{j|mb}\, \mathring{\sigma}_a^m + I_{j,m}\, \mathring{\sigma}_a^m{}_{|b}      
	+ I_{j,b}\, \mathring{\sigma}_a^m{}_{|m} \Big) 
	+ \sum_{j,k} \zeta_{jk}\,   I_{j,m}\, I_{k,b}\, \mathring{\sigma}_a^m \,,
\end{equation}
%-----------------------------
%-----------------------------
\begin{equation}
\begin{aligned}
p_{|ab}^{(\omega)} &=
\omega \left(\mathring{\sigma}_a^n{}_{|mb}\, \mathring{\sigma}_n^m
           + \mathring{\sigma}_a^n{}_{|m}\, \mathring{\sigma}_n^m{}_{|b}
           + \mathring{\sigma}_a^n{}_{|b}\, \mathring{\sigma}_n^m{}_{|m}
           + \mathring{\sigma}_a^n\, \mathring{\sigma}_n^m{}_{|mb}\right) \\
&\quad
+ \sum_j \omega_j\,
 \Big[ I_{j|mb}\, \mathring{\sigma}_a^n\, \mathring{\sigma}_n^m 
   + I_{j,m}\,(\mathring{\sigma}_a^n{}_{|b}\, \mathring{\sigma}_n^m
             + \mathring{\sigma}_a^n\, \mathring{\sigma}_n^m{}_{|b}) 
   + I_{j,b}\,(\mathring{\sigma}_a^n{}_{|m}\, \mathring{\sigma}_n^m
             + \mathring{\sigma}_a^n\, \mathring{\sigma}_n^m{}_{|m}) \Big] \\
&\quad
+ \sum_{j,k} \omega_{jk}\, I_{j,m}\, I_{k,b}\,
   (\mathring{\sigma}_a^n\, \mathring{\sigma}_n^m)\,,
\end{aligned}
\end{equation}
%-----------------------------
and
%-----------------------------
\begin{equation}
\begin{aligned}
p_{|ab}^{(\xi)} &=
\xi \Big(b_a^n{}_{|mb}\, \mathring{\sigma}_n^m
        + b_a^n{}_{|m}\, \mathring{\sigma}_n^m{}_{|b}
        + b_a^n{}_{|b}\, \mathring{\sigma}_n^m{}_{|m}
        + b_a^n\, \mathring{\sigma}_n^m{}_{|mb} \\
&\qquad
        + \mathring{\sigma}_a^n{}_{|mb}\, b_n^m
        + \mathring{\sigma}_a^n{}_{|m}\, b_n^m{}_{|b}
        + \mathring{\sigma}_a^n{}_{|b}\, b_n^m{}_{|m}
        + \mathring{\sigma}_a^n\, b_n^m{}_{|mb}\Big) \\
&\quad
+ \sum_j \xi_j\,
 \Big( I_{j|mb}\,(b_a^n\, \mathring{\sigma}_n^m + \mathring{\sigma}_a^n\, b_n^m) \\
&\qquad
   + I_{j,m}\,(b_a^n{}_{|b}\, \mathring{\sigma}_n^m 
             + b_a^n\, \mathring{\sigma}_n^m{}_{|b} 
             + \mathring{\sigma}_a^n{}_{|b}\, b_n^m 
             + \mathring{\sigma}_a^n\, b_n^m{}_{|b}) \\
&\qquad
   + I_{j,b}\,\left(b_a^n{}_{|m}\, \mathring{\sigma}_n^m 
             + b_a^n\, \mathring{\sigma}_n^m{}_{|m} 
             + \mathring{\sigma}_a^n{}_{|m}\, b_n^m 
             + \mathring{\sigma}_a^n\, b_n^m{}_{|m} \right) \Big) \\
&\quad
+ \sum_{j,k} \xi_{jk}\, I_{j,m}\, I_{k,b}\,
   (b_a^n\, \mathring{\sigma}_n^m + \mathring{\sigma}_a^n\, b_n^m)\,,
\end{aligned}
\end{equation}
%-----------------------------
%-----------------------------
\begin{equation}
\begin{aligned}
p_{|ab}^{(\eta)} &=
\eta\;\big(c_a^n{}_{|mb}\, \mathring{\sigma}_n^m
         + c_a^n{}_{|m}\, \mathring{\sigma}_n^m{}_{|b}
         + c_a^n{}_{|b}\, \mathring{\sigma}_n^m{}_{|m}
         + c_a^n\, \mathring{\sigma}_n^m{}_{|mb} \\
&\qquad
         + \mathring{\sigma}_a^n{}_{|mb}\, c_n^m
         + \mathring{\sigma}_a^n{}_{|m}\, c_n^m{}_{|b}
         + \mathring{\sigma}_a^n{}_{|b}\, c_n^m{}_{|m}
         + \mathring{\sigma}_a^n\, c_n^m{}_{|mb}\big) \\
&\quad
+ \sum_j \eta_j\,
 \Big( I_{j|mb}\,(c_a^n\, \mathring{\sigma}_n^m + \mathring{\sigma}_a^n\, c_n^m) \\
&\qquad
   + I_{j,m}\,(c_a^n{}_{|b}\, \mathring{\sigma}_n^m 
             + c_a^n\, \mathring{\sigma}_n^m{}_{|b} 
             + \mathring{\sigma}_a^n{}_{|b}\, c_n^m 
             + \mathring{\sigma}_a^n\, c_n^m{}_{|b}) \\
&\qquad
   + I_{j,b}\,(c_a^n{}_{|m}\, \mathring{\sigma}_n^m 
             + c_a^n\, \mathring{\sigma}_n^m{}_{|m} 
             + \mathring{\sigma}_a^n{}_{|m}\, c_n^m 
             + \mathring{\sigma}_a^n\, c_n^m{}_{|m}) \Big) \\
&\quad
+ \sum_{j,k} \eta_{jk}\, I_{j,m}\, I_{k,b}\,
   (c_a^n\, \mathring{\sigma}_n^m + \mathring{\sigma}_a^n\, c_n^m)\,,
\end{aligned}
\end{equation}
%-----------------------------
%-----------------------------
\begin{equation}
\begin{aligned}
p_{|ab}^{(\chi)} &=
\chi \Big( b^n_{a|mb}\,\mathring{\sigma}^{2m}_n
        + b^n_{a|m}\,\mathring{\sigma}^{2m}_{~n|b}
        + b^n_{a|b}\,\mathring{\sigma}^{2m}_{~n|m}
        + b_a^n\,\mathring{\sigma}^{2m}_{~n|mb} \\
&\qquad
        + \mathring{\sigma}^{2n}_{~a|mb}\,b_n^m
        + \mathring{\sigma}^{2n}_{~a|m}\,b_n^m{}_{|b}
        + \mathring{\sigma}^{2n}_{~a|b}\,b_n^m{}_{|m}
        + \mathring{\sigma}^{2n}_a\,b_n^m{}_{|mb} \Big) \\
&\quad
+ \sum_j \chi_j \Big[ I_{j|bm}\,\big(b_a^n\,\mathring{\sigma}^{2m}_n + \mathring{\sigma}^{2n}_a\,b_n^m\big) \\
&\qquad
   + I_{j,m}\,\big(b^n_{a|b}\,\mathring{\sigma}^{2m}_n + b_a^n\,\mathring{\sigma}^{2m}_{~n|b}
                 + \mathring{\sigma}^{2n}_{~a|b}\,b_n^m + \mathring{\sigma}^{2n}_a\,b_n^m{}_{|b} \big) \\
&\qquad
   + I_{j,b}\,\big(b^n_{a|m}\,\mathring{\sigma}^{2m}_n + b_a^n\,\mathring{\sigma}^{2m}_{~n|m}
                 + \mathring{\sigma}^{2n}_{~a|m}\,b_n^m + \mathring{\sigma}^{2n}_a\,b_n^m{}_{|m} \big) \Big] \\
&\quad
+ \sum_{j,k} \chi_{jk}\, I_{j,b}\,I_{k,m}\,\big(b_a^n\,\mathring{\sigma}^{2m}_n + \mathring{\sigma}^{2n}_a\,b_n^m\big) \,.
\end{aligned}
\end{equation}
%-----------------------------
Collecting all contributions one obtains
%-----------------------------
\begin{equation}
\begin{aligned}
p_{|ab} &=
  \beta \,\mathcal{B}_{ab}
+ \gamma \,\mathcal{C}_{ab}
+ \zeta \,\mathcal{Z}_{ab}
+ \omega \,\mathcal{W}_{ab}
+ \xi \,\mathcal{X}_{ab}
+ \eta \,\mathcal{H}_{ab}
+ \chi \,\mathcal{K}_{ab}\\
&\quad
+ \sum_j \Big(
     \beta_j \,\mathcal{B}^{(j)}_{ab}
   + \gamma_j \,\mathcal{C}^{(j)}_{ab}
   + \zeta_j \,\mathcal{Z}^{(j)}_{ab}
   + \omega_j \,\mathcal{W}^{(j)}_{ab}
   + \xi_j \,\mathcal{X}^{(j)}_{ab}
   + \eta_j \,\mathcal{H}^{(j)}_{ab}
   + \chi_j \,\mathcal{K}^{(j)}_{ab}
  \Big) \\
&\quad
+ \sum_{j,k} \Big(
     \beta_{jk}\, \mathcal{B}^{(jk)}_{ab}
   + \gamma_{jk}\, \mathcal{C}^{(jk)}_{ab}
   + \zeta_{jk}\, \mathcal{Z}^{(jk)}_{ab}
   + \omega_{jk}\, \mathcal{W}^{(jk)}_{ab}
   + \xi_{jk}\, \mathcal{X}^{(jk)}_{ab}
   + \eta_{jk}\, \mathcal{H}^{(jk)}_{ab}
   + \chi_{jk}\, \mathcal{K}^{(jk)}_{ab}
  \Big)\,.
\end{aligned}
\end{equation}
%-----------------------------
For each scalar coefficient $f \in \{\beta,\gamma,\zeta,\omega,\xi,\eta,\chi\}$, 
the total contribution to $p_{|ab}$ is decomposed into three groups as shown in the following table.
%-----------------------------
\begin{center}
\vskip 0.1in
\renewcommand{\arraystretch}{1.4}
\begin{tabular}{|c|c|c|c|}
\hline
Response function $f$ & $f$--coefficient & $f_j$--coefficients & $f_{jk}$--coefficients \\
\hline
$\beta$ & $\mathcal{B}_{ab}$ & $\mathcal{B}^{(j)}_{ab}$ & $\mathcal{B}^{(jk)}_{ab}$ \\
\hline
$\gamma$ & $\mathcal{C}_{ab}$ & $\mathcal{C}^{(j)}_{ab}$ & $\mathcal{C}^{(jk)}_{ab}$ \\
\hline
$\zeta$ & $\mathcal{Z}_{ab}$ & $\mathcal{Z}^{(j)}_{ab}$ & $\mathcal{Z}^{(jk)}_{ab}$ \\
\hline
$\omega$ & $\mathcal{W}_{ab}$ & $\mathcal{W}^{(j)}_{ab}$ & $\mathcal{W}^{(jk)}_{ab}$ \\
\hline
$\xi$ & $\mathcal{X}_{ab}$ & $\mathcal{X}^{(j)}_{ab}$ & $\mathcal{X}^{(jk)}_{ab}$ \\
\hline
$\eta$ & $\mathcal{H}_{ab}$ & $\mathcal{H}^{(j)}_{ab}$ & $\mathcal{H}^{(jk)}_{ab}$ \\
\hline
$\chi$ & $\mathcal{K}_{ab}$ & $\mathcal{K}^{(j)}_{ab}$ & $\mathcal{K}^{(jk)}_{ab}$ \\
\hline
\end{tabular}
\end{center}
\vskip 0.1in
%-----------------------------
Because the scalar response functions $f\in\{\beta,\gamma,\zeta,\omega,\xi,\eta,\chi\}$ are arbitrary functions of the invariants, the fields $f$, $f_j$, and $f_{jk}$ appear independently in the expression for $p_{|ab}$. The integrability condition $p_{|ab}=p_{|ba}$ therefore must hold separately for the coefficients of $f$, of $f_j$, and of $f_{jk}$.
However, the second derivatives of the response functions satisfy the symmetry $f_{jk}=f_{kj}$. Consequently, the coefficients multiplying $f_{jk}$ are not independent, and only their symmetric combinations are relevant. Thus, the universality constraints associated with second derivatives of the response functions must be imposed on the symmetrized coefficients
$\accentset{s}{\mathcal{A}}^{(jk)}_{ab}\coloneqq\mathcal{A}^{(jk)}_{ab}+\mathcal{A}^{(kj)}_{ab}$, rather than on $\mathcal{A}^{(jk)}_{ab}$ alone.
Accordingly, symmetry in $(a,b)$ must be enforced separately on the coefficients of $f$, on the coefficients of $f_j$, and on the symmetrized coefficients associated with $f_{jk}$.
These coefficients are given explicitly below.
%-----------------------------
\begin{align}
	\label{Universality-C-1}
	\mathcal{B}_{ab} & =  b_a^m{}_{|mb}	\,, \\
	\label{Universality-C-2}
	\mathcal{B}^{(j)}_{ab} &= I_{j|mb}\, b_a^m + I_{j,m}\, b_a^m{}_{|b} 
	+ I_{j,b}\, b_a^m{}_{|m}  \,,\\
	\label{Universality-C-3}
	\accentset{s}{\mathcal{B}}^{(jk)}_{ab} & = I_{j,m}\, I_{k,b}\, b_a^m 
	+ I_{k,m}\, I_{j,b}\, b_a^m  \,,\\
	\label{Universality-C-4}
	\mathcal{C}_{ab} & = c_a^m{}_{|mb} 	\,, \\
	\label{Universality-C-5}
	\mathcal{C}^{(j)}_{ab} &	 = I_{j|mb}\, c_a^m + I_{j,m}\, c_a^m{}_{|b} 
	+ I_{j,b}\, c_a^m{}_{|m} \,,\\
	\label{Universality-C-6}
	\accentset{s}{\mathcal{C}}^{(jk)}_{ab} & =  I_{j,m}\, I_{k,b}\, c_a^m 
	+ I_{k,m}\, I_{j,b}\, c_a^m \,,\\
	\label{Universality-C-7}
	\mathcal{Z}_{ab} & = \mathring{\sigma}_a^m{}_{|mb}	\,, \\
	\mathcal{Z}^{(j)}_{ab} &	 = I_{j|mb}\, \mathring{\sigma}_a^m 
	+ I_{j,m}\, \mathring{\sigma}_a^m{}_{|b}      
	+ I_{j,b}\, \mathring{\sigma}_a^m{}_{|m} \,,\\
	\label{Universality-C-9}
	\accentset{s}{\mathcal{Z}}^{(jk)}_{ab} & = I_{j,m}\, I_{k,b}\, \mathring{\sigma}_a^m 
	+ I_{k,m}\, I_{j,b}\, \mathring{\sigma}_a^m \,,\\
	\mathcal{W}_{ab} & = \mathring{\sigma}_a^n{}_{|mb}\, \mathring{\sigma}_n^m
	+ \mathring{\sigma}_a^n{}_{|m}\, \mathring{\sigma}_n^m{}_{|b}
	+ \mathring{\sigma}_a^n{}_{|b}\, \mathring{\sigma}_n^m{}_{|m}
	+ \mathring{\sigma}_a^n\, \mathring{\sigma}_n^m{}_{|mb}\,, \\
	\mathcal{W}^{(j)}_{ab} & = I_{j|mb}\, \mathring{\sigma}_a^n\, \mathring{\sigma}_n^m 
	+ I_{j,m}\,(\mathring{\sigma}_a^n{}_{|b}\, \mathring{\sigma}_n^m
             + \mathring{\sigma}_a^n\, \mathring{\sigma}_n^m{}_{|b}) 
	+ I_{j,b}\,(\mathring{\sigma}_a^n{}_{|m}\, \mathring{\sigma}_n^m
             + \mathring{\sigma}_a^n\, \mathring{\sigma}_n^m{}_{|m})  \,,\\
	\label{Universality-C-10}
	\accentset{s}{\mathcal{W}}^{(jk)}_{ab} & 
	= I_{j,m}\, I_{k,b}\,\mathring{\sigma}_a^n\, \mathring{\sigma}_n^m
	+I_{k,m}\, I_{j,b}\,\mathring{\sigma}_a^n\, \mathring{\sigma}_n^m\,, \\
	\mathcal{X}_{ab} & = b_a^n{}_{|mb}\, \mathring{\sigma}_n^m
	+ b_a^n{}_{|m}\, \mathring{\sigma}_n^m{}_{|b}
	+ b_a^n{}_{|b}\, \mathring{\sigma}_n^m{}_{|m}
	+ b_a^n\, \mathring{\sigma}_n^m{}_{|mb} 
	+ \mathring{\sigma}_a^n{}_{|mb}\, b_n^m
	+ \mathring{\sigma}_a^n{}_{|m}\, b_n^m{}_{|b} \nonumber \\
	& \quad
	+ \mathring{\sigma}_a^n{}_{|b}\, b_n^m{}_{|m}
	+ \mathring{\sigma}_a^n\, b_n^m{}_{|mb}	\,, \\
	\mathcal{X}^{(j)}_{ab} &	 = I_{j|mb}\,(b_a^n\, \mathring{\sigma}_n^m 
	+ \mathring{\sigma}_a^n\, b_n^m) 
	+ I_{j,m}\,(b_a^n{}_{|b}\, \mathring{\sigma}_n^m 
             + b_a^n\, \mathring{\sigma}_n^m{}_{|b} 
             + \mathring{\sigma}_a^n{}_{|b}\, b_n^m 
             + \mathring{\sigma}_a^n\, b_n^m{}_{|b}) \nonumber\\
	& \quad + I_{j,b}\,(b_a^n{}_{|m}\, \mathring{\sigma}_n^m 
             + b_a^n\, \mathring{\sigma}_n^m{}_{|m} 
             + \mathring{\sigma}_a^n{}_{|m}\, b_n^m 
             + \mathring{\sigma}_a^n\, b_n^m{}_{|m})\,,\\
	\label{Universality-C-12}
	\accentset{s}{\mathcal{X}}^{(jk)}_{ab} & 
	= I_{j,m}\, I_{k,b} \left(b_a^n\, \mathring{\sigma}_n^m 
	+ \mathring{\sigma}_a^n\, b_n^m \right)
	+I_{k,m}\, I_{j,b} \left(b_a^n\, \mathring{\sigma}_n^m 
	+ \mathring{\sigma}_a^n\, b_n^m \right) \,,\\
	\mathcal{H}_{ab} & = c_a^n{}_{|mb}\, \mathring{\sigma}_n^m
	+ c_a^n{}_{|m}\, \mathring{\sigma}_n^m{}_{|b}
	+ c_a^n{}_{|b}\, \mathring{\sigma}_n^m{}_{|m}
	+ c_a^n\, \mathring{\sigma}_n^m{}_{|mb} 
	+ \mathring{\sigma}_a^n{}_{|mb}\, c_n^m
	+ \mathring{\sigma}_a^n{}_{|m}\, c_n^m{}_{|b} \nonumber \\
	& \quad
	+ \mathring{\sigma}_a^n{}_{|b}\, c_n^m{}_{|m}
	+ \mathring{\sigma}_a^n\, c_n^m{}_{|mb}	\,, \\
	\mathcal{H}^{(j)}_{ab} &	 = I_{j|mb} \left(c_a^n\, \mathring{\sigma}_n^m 
	+ \mathring{\sigma}_a^n\, c_n^m \right) 
	+ I_{j,m} \left(c_a^n{}_{|b}\, \mathring{\sigma}_n^m 
             + c_a^n\, \mathring{\sigma}_n^m{}_{|b} 
             + \mathring{\sigma}_a^n{}_{|b}\, c_n^m 
             + \mathring{\sigma}_a^n\, c_n^m{}_{|b} \right) \nonumber\\
	& \quad + I_{j,b} \left(c_a^n{}_{|m}\, \mathring{\sigma}_n^m 
             + c_a^n\, \mathring{\sigma}_n^m{}_{|m} 
             + \mathring{\sigma}_a^n{}_{|m}\, c_n^m 
             + \mathring{\sigma}_a^n\, c_n^m{}_{|m} \right) \,,\\
	\accentset{s}{\mathcal{H}}^{(jk)}_{ab} & 
	= I_{j,m}\, I_{k,b} \left(c_a^n\, \mathring{\sigma}_n^m 
	+ \mathring{\sigma}_a^n\, c_n^m \right)
	+I_{k,m}\, I_{j,b} \left(c_a^n\, \mathring{\sigma}_n^m 
	+ \mathring{\sigma}_a^n\, c_n^m \right)\,,\\
	\mathcal{K}_{ab} & = b^n_{a|mb}\,\mathring{\sigma}^{2m}_n
        + b^n_{a|m}\,\mathring{\sigma}^{2m}_{~n|b}
        + b^n_{a|b}\,\mathring{\sigma}^{2m}_{~n|m}
        + b_a^n\,\mathring{\sigma}^{2m}_{~n|mb} \nonumber\\
	&\quad
        + \mathring{\sigma}^{2n}_{~a|mb}\,b_n^m
        + \mathring{\sigma}^{2n}_{~a|m}\,b_n^m{}_{|b}
        + \mathring{\sigma}^{2n}_{~a|b}\,b_n^m{}_{|m}
        + \mathring{\sigma}^{2n}_a\,b_n^m{}_{|mb} \,, \\
	\mathcal{K}^{(j)}_{ab} & = I_{j|bm}\,\big(b_a^n\,\mathring{\sigma}^{2m}_n 
	+ \mathring{\sigma}^{2n}_a\,b_n^m\big) \nonumber\\
	&\quad
	+ I_{j,m}\,\big(b^n_{a|b}\,\mathring{\sigma}^{2m}_n 
	+ b_a^n\,\mathring{\sigma}^{2m}_{~n|b}
                 + \mathring{\sigma}^{2n}_{~a|b}\,b_n^m 
                 + \mathring{\sigma}^{2n}_a\,b_n^m{}_{|b} \big) \nonumber\\
	&\quad
	+ I_{j,b}\,\big(b^n_{a|m}\,\mathring{\sigma}^{2m}_n 
	+ b_a^n\,\mathring{\sigma}^{2m}_{~n|m}
                 + \mathring{\sigma}^{2n}_{~a|m}\,b_n^m 
                 + \mathring{\sigma}^{2n}_a\,b_n^m{}_{|m} \big)  \,,\\
	\accentset{s}{\mathcal{K}}^{(jk)}_{ab} & = 
	I_{j,b}\,I_{k,m}\,\big(b_a^n\,\mathring{\sigma}^{2m}_n 
	+ \mathring{\sigma}^{2n}_a\,b_n^m\big)
	+I_{k,b}\,I_{j,m}\,\big(b_a^n\,\mathring{\sigma}^{2m}_n 
	+ \mathring{\sigma}^{2n}_a\,b_n^m\big)  \,.
\end{align}
%-----------------------------
These terms must be symmetric in $(ab)$ for $j=1,2,4,\cdots, 10$ ($9$ terms) and $j\leq k =1,2,4,\cdots, 10$.
The $\frac{9\times (9+1)}{2}=45$ independent $(j,k)$ pairs are
%-----------------------------
\begin{equation}
\begin{aligned}
&(1,1), (1,2), (1,4), (1,5), (1,6), (1,7), (1,8), (1,9), (1,10), \\
&(2,2), (2,4), (2,5), (2,6), (2,7), (2,8), (2,9), (2,10), \\
&(4,4), (4,5), (4,6), (4,7), (4,8), (4,9), (4,10), \\
&(5,5), (5,6), (5,7), (5,8), (5,9), (5,10), \\
&(6,6), (6,7), (6,8), (6,9), (6,10), \\
&(7,7), (7,8), (7,9), (7,10), \\
&(8,8), (8,9), (8,10), \\
&(9,9), (9,10), \\
&(10,10)\,.
\end{aligned}
\end{equation}
%-----------------------------

For $j=1,2$, from the universality constraints \eqref{Universality-C-1}-\eqref{Universality-C-6}, one recovers the twelve constraints that were derived in \citep{Yavari2024Cauchy} and were shown to be equivalent to Ericksen's nine universality constraints. It is known that other than isochoric homogeneous deformations, these constraints admit five families of universal deformations. For each known family of universal deformations, the residual stress distributions that together with the universal deformations satisfy the remaining universality constraints are called \textit{universal residual stresses}. 

%-----------------------------
\begin{prop}\label{prop:universal-deformations-residual-stress}
The presence of residual stress does not enlarge the space of universal deformations. In particular, every universal deformation of a residually-stressed elastic body is a universal deformation of the corresponding residual-stress-free case.
\end{prop}
\begin{proof}
The universality constraints \eqref{Universality-C-1}-\eqref{Universality-C-6} for $j=1,2$ are identical to the universality constraints of the residual-stress-free case \citep{Yavari2024Cauchy} and are independent of the residual stress. Therefore, any universal deformation must satisfy exactly the same universality constraints as in the classical theory. Consequently, universal deformations beyond those of the residual-stress-free case do not exist.
For a given universal deformation, one must then determine whether there exist non-trivial residual stress fields that satisfy the remaining universality constraints. If such residual stress fields exist, the corresponding deformation remains universal in the presence of residual stress. 
This will be shown for the six known families of universal deformations. Whether there exist additional universal deformations beyond these known families remains an open problem.
\end{proof}
%-----------------------------

From \eqref{Universality-C-1} and \eqref{Universality-C-7}, symmetry of $b_a^m{}_{|mb}$ and $\mathring{\sigma}_a^m{}_{|mb}$ tells us that there exist scalar functions $\phi$ and $\psi$ such that $b_a^m{}_{|m}=\phi_{,a}$ and $\mathring{\sigma}_a^m{}_{|m}=\psi_{,a}$, i.e., $\operatorname{div}\mathbf{b}^\sharp=\nabla \phi$ and $\operatorname{div}\mathring{\boldsymbol{\sigma}}=\nabla \psi$.

%-----------------------------
\begin{remark}
\citet{Ericksen1954} used the following result in his analysis that we briefly review here. Suppose $\mathbf{u}$ and $\mathbf{v}$ are vectors such that $\mathbf{u}\otimes\mathbf{v}=\mathbf{v}\otimes\mathbf{u}$. Assuming that $\mathbf{v}\neq\mathbf{0}$ (if $\mathbf{v}=\mathbf{0}$, this equality trivially holds) one can write
%-----------------------------
\begin{equation} \label{u-v-identity}
	\mathbf{u}=\frac{\mathbf{u}\cdot\mathbf{v}}{|\mathbf{v}|^2}\mathbf{v}=\lambda\mathbf{v}\,,
\end{equation}
%-----------------------------
i.e., $\mathbf{u}$ and $\mathbf{v}$ must be parallel.
Now suppose that $\mathbf{u}^{\flat}=d\phi$ and $\mathbf{v}^{\flat}=d \psi$, where $\flat$ is the flat operator that gives the $1$-form corresponding to a vector, $d$ is the exterior derivative, and $\phi$ and $\psi$ are scalar fields. Note that $\lambda=\lambda(\mathbf{u},\mathbf{v})=\lambda(\phi,\psi)$, in general. Thus, $\mathbf{u}^{\flat}=d\phi=\lambda(\phi,\psi)\, d\psi$.
On any region where $d\psi\neq 0$, this implies that the gradients of $\phi$ and $\psi$ are everywhere parallel. Equivalently, $\phi$ is constant along the level sets of $\psi$, and hence $\phi$ and $\psi$ are locally functionally dependent. Therefore, there exists a scalar function $\Phi$ such that $\phi=\Phi(\psi)$ locally, and consequently
%-----------------------------
\begin{equation}
	d\phi = \Phi'(\psi)\, d\psi\,.
\end{equation}
%-----------------------------
Comparing with $d\phi=\lambda(\phi,\psi)\, d\psi$, we conclude that $\lambda=\Phi'(\psi)$, i.e., $\lambda$ depends only on $\psi$. This implies that in \eqref{u-v-identity}, $\lambda=\lambda(\mathbf{v})$.
\end{remark}
%-----------------------------

Symmetry of $\accentset{s}{\mathcal{B}}^{(jk)}_{ab}$ in \eqref{Universality-C-3} implies that 
%-----------------------------
\begin{equation} \label{Bs-symmetry-constraints}
	(b_a^m\,I_{j,m})\, I_{k,b} + (b_a^m\,I_{k,m})\, I_{j,b}  
	= (b_b^m\,I_{j,m})\, I_{k,a} + (b_b^m\,I_{k,m})\, I_{j,a}\,.
\end{equation}
%-----------------------------
In particular, for $k=j$, one has
%-----------------------------
\begin{equation}
	(b_a^m\,I_{j,m})\, I_{j,b}  = (b_b^m\,I_{j,m})\, I_{j,a} \,,
\end{equation}
%-----------------------------
and hence
%-----------------------------
\begin{equation} \label{b-eigenvalues}
	b_a^m\,I_{j,m} = \xi_{j}\,I_{j,a}\quad (\text{no~summation~on~}j)\,.
\end{equation}
%-----------------------------
Therefore, if $I_j$ is not constant, $\nabla I_j$ is an eigenvector of $\mathbf{b}$ (and obviously of $\mathbf{c}$).\footnote{Ericksen's universal solutions satisfy this property for $j=1, 2$. In the presence of residual stress, we have the same property for $j=4,\cdots,10$ as well.}
Now, substituting \eqref{b-eigenvalues} into \eqref{Bs-symmetry-constraints} for $j\neq k$ one obtains
%-----------------------------
\begin{equation} 
	(\xi_j - \xi_k) I_{j,a}\,I_{k,b} = (\xi_j - \xi_k) I_{j,b}\,I_{k,a}\,.
\end{equation}
%-----------------------------
If $\xi_j \neq \xi_k$, this implies that $\nabla I_j$ and $\nabla I_j$ are parallel, i.e., $I_{k,a} = \xi_{kj}\,I_{j,a}$ (no summation on $j$). If $\xi_j = \xi_k$, $\nabla I_j$ and $\nabla I_j$ are still parallel (eigenvectors corresponding to the same eigenvalue).

We know that $\mathbf{b}$ and $\mathbf{c}$ have the same eigenvectors. Therefore, knowing that $\nabla I_i$ is an eigenvector of $\mathbf{b}$, the universality constraints corresponding to $\accentset{s}{\mathcal{C}}^{(jk)}_{ab}$ are trivially satisfied.

Symmetry of $\accentset{s}{\mathcal{Z}}^{(jk)}_{ab} $ in \eqref{Universality-C-9} implies that 
%-----------------------------
\begin{equation}
	(\mathring{\sigma}_a^m\,I_{j,m})\, I_{k,b} + (\mathring{\sigma}_a^m\,I_{k,m})\, I_{j,b}  
	= (\mathring{\sigma}_b^m\,I_{j,m})\, I_{k,a} + (\mathring{\sigma}_b^m\,I_{k,m})\, I_{j,a}\,.
\end{equation}
%-----------------------------
In particular, for $k=j$, one has
%-----------------------------
\begin{equation}
	(\mathring{\sigma}_a^m\,I_{j,m})\, I_{j,b}  = (\mathring{\sigma}_b^m\,I_{j,m})\, I_{j,a} \,,
\end{equation}
%-----------------------------
and hence
%-----------------------------
\begin{equation} \label{sigma0-eigenvalues}
	\mathring{\sigma}_a^m\,I_{j,m} = \eta_{j}\,I_{j,a}\quad (\text{no~summation~on~}j)\,.
\end{equation}
%-----------------------------
This implies that if $I_j$ is not constant, $\nabla I_j$ is an eigenvector of $\mathring{\boldsymbol{\sigma}}$, i.e.,
%-----------------------------
\begin{equation}
	\mathring{\sigma}_a^m\,I_{j,m} = \eta_{j}\,I_{j,a}\quad (\text{no~summation~on~}j)\,.
\end{equation}
%-----------------------------
For $k \neq j$, we have
%-----------------------------
\begin{equation}
	(\eta_j - \eta_k) I_{j,a}\,I_{k,b} = (\eta_j - \eta_k) I_{j,b}\,I_{k,a}\,.
\end{equation}
%-----------------------------
If $\eta_j = \eta_k$, this identity is satisfied trivially. If $\eta_j \neq \eta_k$, the factor $(\eta_j - \eta_k)$ can be canceled, and the relation reduces to
$ I_{j,a}\,I_{k,b} = I_{j,b}\,I_{k,a} $, which trivially holds because $\nabla I_j \parallel \nabla I_k$.

The terms $\accentset{s}{\mathcal{W}}^{(jk)}_{ab}$ in \eqref{Universality-C-10} are now simplified to read
%-----------------------------
\begin{equation}
	\accentset{s}{\mathcal{W}}^{(jk)}_{ab} 
	= 2\left( \eta_j^2\, I_{j,a}\, I_{k,b} + \eta_k^2\, I_{j,b}\, I_{k,a} \right)
	\,.
\end{equation}
%-----------------------------
The corresponding universality constraints are: $(\eta_j^2-\eta_k^2)\, I_{j,a}\, I_{k,b}=(\eta_j^2-\eta_k^2)\, I_{j,b}\, I_{k,a}$.
If $\eta_j = \eta_k$, the constraints are trivially satisfied.  
If $\eta_j \neq \eta_k$, the constraints are still satisfied since $\nabla I_j \parallel \nabla I_k$.

The terms $\accentset{s}{\mathcal{X}}^{(jk)}_{ab} $ in \eqref{Universality-C-12} are now simplified to read
%-----------------------------
\begin{equation}
\begin{aligned}
	\accentset{s}{\mathcal{X}}^{(jk)}_{ab} 
	 = I_{j,m}\, I_{k,b} \left(b_a^n\, \mathring{\sigma}_n^m 
	+ \mathring{\sigma}_a^n\, b_n^m \right)
	+I_{k,m}\, I_{j,b} \left(b_a^n\, \mathring{\sigma}_n^m 
	+ \mathring{\sigma}_a^n\, b_n^m \right) 
	 = 2\left( 	\eta_j\,\xi_j\, I_{j,a}\, I_{k,b} + \eta_k\,\xi_k\, I_{j,b}\, I_{k,a} \right)
	\,.
\end{aligned}
\end{equation}
%-----------------------------
The corresponding universality constraints are: $(\eta_j\,\xi_j-\eta_k\,\xi_k)\, I_{j,b}\, I_{k,a}=(\eta_j\,\xi_j-\eta_k\,\xi_k)\, I_{j,a}\, I_{k,b}$.
If $\eta_j\,\xi_j = \eta_k\,\xi_k$, the constraints are trivially satisfied.  
If $\eta_j\,\xi_j \neq \eta_k\,\xi_k$, the constraints are still trivially satisfied since $\nabla I_j \parallel \nabla I_k$.

Next we simplify the terms $\accentset{s}{\mathcal{H}}^{(jk)}_{ab}$ as
%-----------------------------
\begin{equation}
	\accentset{s}{\mathcal{H}}^{(jk)}_{ab} 
	= 2\left( \frac{\eta_j}{\xi_j} \,I_{j,a}\, I_{k,b} + \frac{\eta_k}{\xi_k}\, I_{j,b}\, I_{k,a} \right)
	\,.
\end{equation}
%-----------------------------
The corresponding universality constraints read
%-----------------------------
\begin{equation}
	\left(\frac{\eta_j}{\xi_j}-\frac{\eta_k}{\xi_k}\right)\, I_{j,a}\, I_{k,b}
	= \left(\frac{\eta_j}{\xi_j}-\frac{\eta_k}{\xi_k}\right)\, I_{j,b}\, I_{k,a}	
	\,.
\end{equation}
%-----------------------------
If $\eta_j\,\xi_j^{-1} = \eta_k\,\xi_k^{-1}$, the constraints are trivially satisfied.  
If $\eta_j\,\xi_j^{-1} \neq \eta_k\,\xi_k^{-1}$, the constraints are still satisfied since $\nabla I_j \parallel \nabla I_k$.

The terms $\accentset{s}{\mathcal{K}}^{(jk)}_{ab}$ is simplified as
%-----------------------------
\begin{equation}
	\accentset{s}{\mathcal{K}}^{(jk)}_{ab} 
	= 2\left( \xi_j\,\eta_j^2 \,I_{j,a}\, I_{k,b} + \xi_k\,\eta_k^2\, I_{j,b}\, I_{k,a} \right)
	\,.
\end{equation}
%-----------------------------
The corresponding universality constraints read: $\left(\xi_j\,\eta_j^2-\xi_k\,\eta_k^2\right)\, I_{j,a}\, I_{k,b}	= \left(\xi_j\,\eta_j^2-\xi_k\,\eta_k^2\right)\, I_{j,b}\, I_{k,a}$.
If $\xi_j\,\eta_j^2=\xi_k\,\eta_k^2$, the constraints are trivially satisfied.  
If $\xi_j\,\eta_j^2\neq \xi_k\,\eta_k^2$, the constraints are still satisfied since $\nabla I_j \parallel \nabla I_k$.

Recall that $\mathring{\boldsymbol{\sigma}}$ has principal invariants $I_4$, $I_5$, and $I_6$. If they are not constant, we have shown that $\nabla I_4$, $\nabla I_5$, and $\nabla I_6$ are eigenvectors of $\mathbf{b}$ and are parallel. In particular, $I_4$, $I_5$, and $I_6$ are functionally dependent.

%-----------------------------
%----------------------------
\subsection{Symmetry Reduction of the system of nonlinear universality PDEs}

\citet{Goodbrake2020} studied a closely related problem in the context of universal eigenstrains. In their formulation, the material manifold is endowed with a Riemannian metric $\mathbf{G}$ that encodes the eigenstrain distribution and is not known a priori. Two key assumptions were made in that work. First, the analysis was restricted to the six classical families of universal deformations known from incompressible isotropic elasticity. Second, it was assumed that the material metric $\mathbf{G}$ has the same symmetry group as the right Cauchy--Green tensor $\mathbf{C}^\flat$, i.e., the eigenstrain-induced geometry inherits the symmetries of the deformation.

In the present work, we adopt a fundamentally different point of view by working directly with residual stresses rather than eigenstrains or an underlying material metric. In particular, we do not assume \emph{a priori} that the deformation belongs to any of the known families of universal deformations. Instead, we have shown that for incompressible isotropic Cauchy elastic solids with residual stress, no additional universal deformations exist beyond those of the residual-stress-free case (Proposition~\ref{prop:universal-deformations-residual-stress}). This constitutes an important structural result: residual stress does not enlarge the class of universal deformations.

However, even for homogeneous deformations, solving the full system of nonlinear PDEs arising from the universality constraints is highly nontrivial. The complexity is significantly greater than in the classical setting due to the coupling between the deformation and the residual stress field. To render the problem tractable, we follow an approach analogous in spirit to that of \citet{Goodbrake2020}: for each family of universal deformations, we restrict attention to residual stress fields that share the same symmetry group as the corresponding deformation. In other words, the residual stress is assumed to be invariant under the symmetry group associated with the universal deformation under consideration, i.e., $\mathring{\mathbf{S}}$ has the same symmetry as $\mathbf{C}^\flat$.
This symmetry restriction is not imposed as a fundamental constitutive assumption, but rather as a pragmatic simplification that allows the universality constraints to be analyzed explicitly.

%----------------------
%----------------------
\subsection{Group actions and symmetry}

Symmetry will be used below to reduce the universality PDEs. Recall that the special Euclidean group is the semi-direct product $\mathrm{SE}(n)=\mathrm{SO}(n)\ltimes \mathrm{T}(n)$, where $\mathrm{T}(n)$ is the group of translations and $\mathrm{SO}(n)$ is the group of proper rotations. An element of $\mathrm{SE}(n)$ is written as $(\mathbf{Q},\mathbf{c})$, where $\mathbf{Q}\in\mathrm{SO}(n)$ and $\mathbf{c}\in\mathrm{T}(n)$. Its action on Euclidean space is given as $\mathbf{x}\mapsto \mathbf{Q}\mathbf{x}+\mathbf{c}$.
Therefore, the group product is defined as
%-----------------------------
\begin{equation}
	(\mathbf{Q}_2,\mathbf{c}_2)\star(\mathbf{Q}_1,\mathbf{c}_1)
	= (\mathbf{Q}_2\mathbf{Q}_1,\mathbf{Q}_2\mathbf{c}_1+\mathbf{c}_2)\,.
\end{equation}
%-----------------------------
More generally, an action of a group $\mathcal{G}$ on a manifold $\mathcal{M}$ is a map $\rho:\mathcal{G}\times\mathcal{M}\to\mathcal{M}$ such that, writing $\rho(g,\mathbf{x})=g\cdot\mathbf{x}$, one has
%-----------------------------
\begin{equation}
	g_2\cdot(g_1\cdot\mathbf{x})=(g_2 g_1)\cdot\mathbf{x}\,,\qquad
	e\cdot\mathbf{x}=\mathbf{x}\,,
\end{equation}
%-----------------------------
for all $g_1,g_2\in \mathcal{G}$ and all $\mathbf{x}\in\mathcal{M}$, where $e$ is the identity element of $\mathcal{G}$.

The action can be prolonged to tensor bundles. For a fixed element $g\in \mathcal{G}$, the map $\rho(g):\mathcal{M}\to\mathcal{M}$ is a diffeomorphism. Its tangent map defines the induced action on tangent vectors, while the pull-back by the inverse map defines the induced action on covectors. Thus, tensor fields are transformed by acting both on the base point and on the tensor components. This distinction is important because an action may fix a point while still acting nontrivially on the tangent space at that point. For example, a rotation about an axis fixes the points on the axis but rotates tangent vectors at those points.

For the classical universal deformations reviewed below, the right Cauchy--Green tensor $\mathbf{C}^\flat$ is invariant under the prolonged action of a Lie subgroup of $\mathrm{SE}(3)$ acting on the reference configuration. This invariance will be used in the symmetry reduction of the system of nonlinear universality PDEs.

%----------------------
%----------------------
\subsection{Prolonged action on $\mathbf{C}^\flat$}

Let a Lie group $\mathcal{G}$ act on the body manifold $\mathcal{B}$ via $\rho:\mathcal{G}\times\mathcal{B}\to\mathcal{B}$, $(g,X)\mapsto g\cdot X$. For each $g\in\mathcal{G}$, the map $\rho_g:\mathcal{B}\to\mathcal{B}$ is a diffeomorphism with tangent map
%-----------------------------
\begin{equation}
	T_X\rho_g:T_X\mathcal{B}\rightarrow T_{g\cdot X}\mathcal{B}\,.
\end{equation}
%-----------------------------
The action of $\mathcal{G}$ on points induces an action on tensor fields. For a $(0,2)$-tensor field $\mathbf{T}$, the values $\mathbf{T}_X$ and $\mathbf{T}_{g\cdot X}$ cannot be compared directly, since they are bilinear forms on different tangent spaces. 
The transformed tensor should again be a tensor field on $\mathcal{B}$ and hence must assign a tensor to each point $X\in\mathcal{B}$. Since the group action maps points according to $X\mapsto g\cdot X$, the value of the transformed tensor at $X$ is obtained from the value of the original tensor at the point $g^{-1}\cdot X$. The pull-back by $\rho_{g^{-1}}$ converts this tensor into a bilinear form on $T_X\mathcal{B}$.\footnote{The pull-back is used because $\mathbf{T}$ acts on vectors. Pulling back $\mathbf{T}$ allows it to act on vectors based at the original point. In contrast, vector fields are transformed by push-forward.}
This defines the prolonged action
%-----------------------------
\begin{equation}
	(g\cdot\mathbf{T})_X	= (\rho_{g^{-1}})^*\mathbf{T}\big|_X\,.
\end{equation}
%-----------------------------
Equivalently, for $\mathbf{U},\mathbf{V}\in T_X\mathcal{B}$,
%-----------------------------
\begin{equation}
	(g\cdot\mathbf{T})_X(\mathbf{U},\mathbf{V})
	=\mathbf{T}_{g^{-1}\cdot X}
	\left(T_X\rho_{g^{-1}}\,\mathbf{U},T_X\rho_{g^{-1}}\,\mathbf{V}\right)\,.
\end{equation}
%-----------------------------
The use of $g^{-1}$ is the convention that makes this a left action on tensor fields, namely
%-----------------------------
\begin{equation}
	g_2\cdot(g_1\cdot\mathbf{T}) =	(g_2g_1)\cdot\mathbf{T}\,.
\end{equation}
%-----------------------------
Applying this to the right Cauchy--Green tensor $\mathbf{C}^\flat$, one obtains
%-----------------------------
\begin{equation}
	(g\cdot\mathbf{C}^\flat)_X=(\rho_{g^{-1}})^*\mathbf{C}^\flat\big|_X\,.
\end{equation}
%-----------------------------
The tensor field $\mathbf{C}^\flat$ is invariant under the action of $\mathcal{G}$ if
%-----------------------------
\begin{equation}
	g\cdot\mathbf{C}^\flat = \mathbf{C}^\flat\,,
	\qquad	\forall\,g\in\mathcal{G}\,.
\end{equation}
%-----------------------------

%-----------------------------
%-----------------------------
\subsection{The known universal deformations and their symmetry groups}

In this section, we recall the classical families of universal deformations in the absence of residual stress \citep{Ericksen1954, Ericksen1955, SinghPipkin1965, KlingbeilShield1966}. Our emphasis is instead on the symmetry properties of these deformations \citep{Goodbrake2020}, which will be central to the subsequent analysis.
Using standard Cartesian coordinates $\{X,Y,Z\}$, $\{x,y,z\}$, cylindrical coordinates $\{R,\Theta,Z\}$, $\{r,\theta,z\}$, and spherical coordinates $\{R,\Theta,\Phi\}$, $\{r,\theta,\varphi\}$, with capital letters denoting reference coordinates and lower case letters denoting spatial coordinates, the known universal deformations fall into the following families.

%-----------------------------
\paragraph{Family 0: Homogeneous deformations.}
In Cartesian coordinates $\{X^A\}$ and $\{x^a\}$, these deformations are given by
%-----------------------------
\begin{equation}\label{eqn:Family0-u}
	x^a=F^a{}_A X^A + c^a\,,
\end{equation}
%-----------------------------
where $F^a{}_A$ is constant with $\det \mathbf{F}=1$, and $c^a$ is constant. The deformation gradient is uniform, i.e., $F^a{}_A(X)=F^a{}_A$, and hence the right Cauchy--Green tensor $\mathbf{C}^\flat$ is constant.

Consider the action of the translation group $\mathrm{T}(3)$ on $\mathcal{B}$,
%-----------------------------
\begin{equation}
	X^A\mapsto \bar{X}^A = X^A + C^A\,, \qquad C^A\in\mathbb{R}\,.
\end{equation}
%-----------------------------
For this action, the inverse map is $X^A\mapsto X^A-C^A$, and its tangent map is the identity, i.e., $T_X\rho_{g^{-1}}=\mathbf{I}$.
Therefore, the prolonged action on a $(0,2)$-tensor field $\mathbf{T}$ reduces to $(g\cdot\mathbf{T})_X = \mathbf{T}_{X - C}$. Applying this to $\mathbf{C}^\flat$, one obtains $(g\cdot \mathbf{C}^\flat)_X =\mathbf{C}^\flat_{X - C}$.
Since $\mathbf{C}^\flat$ is constant, one has $\mathbf{C}^\flat_{X - C}=\mathbf{C}^\flat_X$, and hence
%-----------------------------
\begin{equation}
	(g\cdot \mathbf{C}^\flat)_X = \mathbf{C}^\flat_X\,,
\end{equation}
%-----------------------------
which shows that $\mathbf{C}^\flat$ is invariant under translations. This corresponds to the natural action of $\mathrm{T}(3)\subset \mathrm{SE}(3)$ on $\mathbb{E}^3$. Here, $\mathbb{E}^3$ denotes three-dimensional Euclidean space.

%-----------------------------
%-----------------------------
\paragraph{Family 1: Bending, stretching, and shearing of a rectangular block.}
Using cylindrical coordinates in the spatial configuration and Cartesian coordinates in the reference configuration, the deformation is written as
%-----------------------------
\begin{equation}\label{eqn:Family1-u}
	r=\sqrt{C_1(2X+C_4)}\,, \qquad
	\theta=C_2(Y+C_5)\,, \qquad
	z=\frac{Z}{C_1C_2}-C_2C_3 Y+C_6\,.
\end{equation}
%-----------------------------
The parameters $C_5$ and $C_6$ correspond to rigid motions and can be disregarded. The deformation gradient reads
%-----------------------------
\begin{equation}\label{eqn:Family1-F}
	\mathbf{F}
	=\frac{C_1}{r}\mathbf{e}_r\otimes \mathbf{E}_X
	+C_2 r\,\mathbf{e}_\theta\otimes \mathbf{E}_Y
	-C_2C_3\,\mathbf{e}_z\otimes \mathbf{E}_Y
	+\frac{1}{C_1C_2}\mathbf{e}_z\otimes \mathbf{E}_Z\,,
\end{equation}
%-----------------------------
and in coordinates
%-----------------------------
\begin{equation}
	\left[F^a{}_A\right]
	=
	\begin{bmatrix}
	\frac{C_1}{r}&0&0\\
	0&C_2&0\\
	0&-C_2C_3&\frac{1}{C_1C_2}
	\end{bmatrix}.
\end{equation}
%-----------------------------
The associated tensor $\mathbf{C}^{\flat}$ has the representation
%-----------------------------
\begin{equation} 
	[C_{AB}]=\begin{bmatrix}
	\frac{C_1}{2X+C_4} & 0 & 0\\
	0 & C_2^2\left[C_1(2X+C_4)+C_3^2\right] & -\frac{C_3}{C_1} \\
	0 & -\frac{C_3}{C_1} & \frac{1}{C_1^2C_2^2}
	\end{bmatrix}
	\,,
\end{equation}
%-----------------------------
which depends only on $X$. Consider the action of the translation group $\mathrm{T}(2)$ on $\mathcal{B}$,
%-----------------------------
\begin{equation}
	Y\mapsto \bar{Y}=Y+D_1\,, \qquad Z\mapsto \bar{Z}=Z+D_2\,.
\end{equation}
%-----------------------------
For this action, the inverse map is $Y\mapsto Y-D_1$, $Z\mapsto Z-D_2$, and the tangent map is the identity, i.e., $T_X\rho_{g^{-1}}=\mathbf{I}$. Therefore, the prolonged action on a $(0,2)$-tensor field $\mathbf{T}$ reduces to $(g\cdot\mathbf{T})_X=\mathbf{T}_{(X,\,Y-D_1,\,Z-D_2)}$. Applying this to $\mathbf{C}^\flat$, one obtains $(g\cdot\mathbf{C}^\flat)_X=\mathbf{C}^\flat_{(X,\,Y-D_1,\,Z-D_2)}$.
Since $\mathbf{C}^\flat$ depends only on $X$, one has $\mathbf{C}^\flat_{(X,\,Y-D_1,\,Z-D_2)}=\mathbf{C}^\flat_X$, and hence
%-----------------------------
\begin{equation}
	(g\cdot \mathbf{C}^\flat)_X=\mathbf{C}^\flat_X\,,
\end{equation}
%-----------------------------
which shows that $\mathbf{C}^\flat$ is invariant under translations in $Y$ and $Z$. This corresponds to the action of $\mathrm{T}(2)\subset \mathrm{SE}(3)$.

The tensor $\mathbf{b}^{\sharp}$ has the representation
%-----------------------------
\begin{equation} 
	[b^{ab}]=\begin{bmatrix}
	\frac{C_1}{C_4+2 X} & 0 & 0 \\
	0 & C_2^2 & -C_2^2 C_3 \\
	0 & -C_2^2 C_3 & \frac{1}{C_1^2 C_2^2}+C_2^2 C_3^2
	\end{bmatrix}
	\,,
\end{equation}
%-----------------------------
which, when expressed in spatial coordinates, depends only on $r$. Consider the action of $\mathrm{SO}(2)\times \mathrm{T}(1)$ on the spatial configuration,
%-----------------------------
\begin{equation}
	\theta\mapsto \bar{\theta}=\theta+\theta_0\,, \qquad z\mapsto \bar{z}=z+z_0\,.
\end{equation}
%-----------------------------
For this action, the inverse map is $\theta\mapsto \theta-\theta_0$, $z\mapsto z-z_0$, and the tangent map preserves the orthonormal frame $\{\mathbf{e}_r,\mathbf{e}_\theta,\mathbf{e}_z\}$. Hence, the prolonged action reduces to $(g\cdot\mathbf{b}^{\sharp})_{\mathbf{x}}=\mathbf{b}^{\sharp}_{(r,\,\theta-\theta_0,\,z-z_0)}$. Since $\mathbf{b}^{\sharp}$ depends only on $r$, one has $\mathbf{b}^{\sharp}_{(r,\,\theta-\theta_0,\,z-z_0)}=\mathbf{b}^{\sharp}_{\mathbf{x}}$, and therefore
%-----------------------------
\begin{equation}
	(g\cdot \mathbf{b}^{\sharp})_{\mathbf{x}}=\mathbf{b}^{\sharp}_{\mathbf{x}}\,,
\end{equation}
%-----------------------------
which shows that $\mathbf{b}^{\sharp}$ is invariant under the action of $\mathrm{SO}(2)\times \mathrm{T}(1)\subset \mathrm{SE}(3)$.

%-----------------------------
%-----------------------------
\paragraph{Family 2: Straightening, stretching, and shearing of a sector of a cylindrical shell.}
In cylindrical coordinates in the reference configuration and Cartesian coordinates in the spatial configuration,
%-----------------------------
\begin{equation}\label{eqn:Family2-u}
	x=\frac{1}{2}C_1 C_2^2 R^2+C_4\,, \qquad
	y=\frac{\Theta}{C_1 C_2}+C_5\,, \qquad
	z=\frac{Z}{C_2}+\frac{C_3\Theta}{C_1 C_2}+C_6\,.
\end{equation}
%-----------------------------
The deformation gradient is written as
%-----------------------------
\begin{equation}\label{eqn:Family2-F}
	\mathbf{F}
	=
	C_1 C_2^2 R\,\mathbf{e}_x\otimes \mathbf{E}_R
	+\frac{1}{C_1 C_2 R}\mathbf{e}_y\otimes \mathbf{E}_\Theta
	+\frac{C_3}{C_1 C_2 R}\mathbf{e}_z\otimes \mathbf{E}_\Theta
	+\frac{1}{C_2}\mathbf{e}_z\otimes \mathbf{E}_Z\,,
\end{equation}
%-----------------------------
and in coordinates
%-----------------------------
\begin{equation}
	\left[F^a{}_A\right]
	=
	\begin{bmatrix}
	C_1 C_2^2 R & 0 & 0 \\
	0 & \frac{1}{C_1 C_2 R} & 0 \\
	0 & \frac{C_3}{C_1 C_2 R} & \frac{1}{C_2}
	\end{bmatrix}.
\end{equation}
%-----------------------------
The associated tensors $\mathbf{C}^{\flat}$ and $\mathbf{b}^{\sharp}$ have the following representations
%-----------------------------
\begin{equation} \label{C-Family2}
	[C_{AB}]=\begin{bmatrix}
	C_1^2 C_2^4 R^2 & 0 & 0 \\
	 0 & \frac{C_3^2+1}{C_1^2 C_2^2} & \frac{C_3}{C_1 C_2^2}   \\
	 0 & \frac{C_3}{C_1 C_2^2} & \frac{1}{C_2^2}
	\end{bmatrix}
	\,,\qquad
	[b^{ab}]=\begin{bmatrix}
	C_1^2 C_2^4 R^2 & 0 & 0 \\
	0 & \frac{1}{C_1^2 C_2^2 R^2} & \frac{C_3}{C_1^2 C_2^2 R^2}   \\
	0 & \frac{C_3}{C_1^2 C_2^2 R^2} & \frac{1}{C_2^2}+\frac{C_3^2}{C_1^2C_2^2 R^2}
	\end{bmatrix}
	\,.
\end{equation}
%-----------------------------
The parameters $C_4$, $C_5$, and $C_6$ correspond to rigid motions and can be disregarded. The tensor $\mathbf{C}^{\flat}$ depends only on $R$. Consider the action of $\mathrm{SO}(2)\times \mathrm{T}(1)$ on $\mathcal{B}$,
%-----------------------------
\begin{equation}
	\Theta\mapsto \bar{\Theta}=\Theta+\Theta_0\,, \qquad
	Z\mapsto \bar{Z}=Z+Z_0\,.
\end{equation}
%-----------------------------
For this action, the inverse map is $\Theta\mapsto \Theta-\Theta_0$, $Z\mapsto Z-Z_0$, and the tangent map preserves the orthonormal frame $\{\mathbf{E}_R,\mathbf{E}_\Theta,\mathbf{E}_Z\}$. Hence, the prolonged action reduces to $(g\cdot\mathbf{T})_X=\mathbf{T}_{(R,\,\Theta-\Theta_0,\,Z-Z_0)}$. Applying this to $\mathbf{C}^\flat$, one obtains $(g\cdot \mathbf{C}^\flat)_X=\mathbf{C}^\flat_{(R,\,\Theta-\Theta_0,\,Z-Z_0)}$.
Since $\mathbf{C}^\flat$ depends only on $R$, one has $\mathbf{C}^\flat_{(R,\,\Theta-\Theta_0,\,Z-Z_0)}=\mathbf{C}^\flat_X$, and hence
%-----------------------------
\begin{equation}
	(g\cdot \mathbf{C}^\flat)_X=\mathbf{C}^\flat_X\,,
\end{equation}
%-----------------------------
which shows that $\mathbf{C}^\flat$ is invariant under the action of $\mathrm{SO}(2)\times \mathrm{T}(1)\subset \mathrm{SE}(3)$.

The tensor $\mathbf{b}^{\sharp}$, when expressed in spatial coordinates, depends only on $x$. Consider the action of $\mathrm{T}(2)$ on the spatial configuration,
%-----------------------------
\begin{equation}
	y\mapsto \bar{y}=y+y_0\,, \qquad z\mapsto \bar{z}=z+z_0\,.
\end{equation}
%-----------------------------
For this action, the inverse map is $y\mapsto y-y_0$, $z\mapsto z-z_0$, and the tangent map is the identity, i.e., $T_{\mathbf{x}}\rho_{g^{-1}}=\mathbf{I}$. Hence, the prolonged action reduces to $(g\cdot\mathbf{b}^{\sharp})_{\mathbf{x}}=\mathbf{b}^{\sharp}_{(x,\,y-y_0,\,z-z_0)}$. Since $\mathbf{b}^{\sharp}$ depends only on $x$, one has $\mathbf{b}^{\sharp}_{(x,\,y-y_0,\,z-z_0)}=\mathbf{b}^{\sharp}_{\mathbf{x}}$, and therefore
%-----------------------------
\begin{equation}
	(g\cdot \mathbf{b}^{\sharp})_{\mathbf{x}}=\mathbf{b}^{\sharp}_{\mathbf{x}}\,,
\end{equation}
%-----------------------------
which shows that $\mathbf{b}^{\sharp}$ is invariant under the action of $\mathrm{T}(2)\subset \mathrm{SE}(3)$.

%-----------------------------
%-----------------------------
\paragraph{Family 3: Inflation, bending, torsion, extension, and shearing of a sector of an annular wedge.}
In cylindrical coordinates in the reference and spatial configurations, the deformation is written as
%-----------------------------
\begin{equation}\label{eqn:Family3-u}
	r=\sqrt{\frac{R^2}{C_1C_4-C_2C_3}+C_5}\,,\qquad
	\theta=C_1\Theta+C_2 Z+C_6\,,\qquad
	z=C_3\Theta+C_4 Z+C_7\,.
\end{equation}
%-----------------------------
Let $K=C_1C_4-C_2C_3\neq 0$. The parameters $C_6$ and $C_7$ correspond to rigid motions and can be disregarded. The deformation gradient reads
%-----------------------------
\begin{equation}\label{eqn:Family3-F}
	\mathbf{F}
	=
	\frac{R}{Kr}\mathbf{e}_r\otimes \mathbf{E}_R
	+\frac{C_1 r}{R}\mathbf{e}_\theta\otimes \mathbf{E}_\Theta
	+C_2 r\,\mathbf{e}_\theta\otimes \mathbf{E}_Z
	+\frac{C_3}{R}\mathbf{e}_z\otimes \mathbf{E}_\Theta
	+C_4\,\mathbf{e}_z\otimes \mathbf{E}_Z\,,
\end{equation}
%-----------------------------
and in coordinates
%-----------------------------
\begin{equation}
	\left[F^a{}_A\right]
	=
	\begin{bmatrix}
	\frac{R}{Kr} & 0 & 0\\
	0 & C_1 & C_2\\
	0 & C_3 & C_4
	\end{bmatrix}\,.
\end{equation}
%-----------------------------
The associated tensor $\mathbf{C}^{\flat}$ has the representation
%-----------------------------
\begin{equation} 
	[C_{AB}]=
	\begin{bmatrix}
	\frac{R^2}{K(K C_5+R^2)} & 0 & 0 \\
	 0 & C_3^2+C_1^2 \left(\frac{R^2}{K}+C_5\right) 
	 & C_1 C_2 \left(\frac{R^2}{K}+C_5\right)+C_3 C_4 \\
	 0 & C_1 C_2 \left(\frac{R^2}{K}+C_5\right)+C_3 C_4 
	 & C_4^2+C_2^2 \left(\frac{R^2}{K}+C_5\right)
	\end{bmatrix}\,.
\end{equation}
%-----------------------------
The tensor $\mathbf{C}^{\flat}$ depends only on $R$. Consider the action of $\mathrm{SO}(2)\times \mathrm{T}(1)$ on $\mathcal{B}$,
%-----------------------------
\begin{equation}
	\Theta\mapsto \bar{\Theta}=\Theta+\Theta_0\,,\qquad Z\mapsto \bar{Z}=Z+Z_0\,.
\end{equation}
%-----------------------------
For this action, the inverse map is $\Theta\mapsto \Theta-\Theta_0$, $Z\mapsto Z-Z_0$, and the tangent map preserves the orthonormal frame $\{\mathbf{E}_R,\mathbf{E}_\Theta,\mathbf{E}_Z\}$. Hence, the prolonged action reduces to $(g\cdot\mathbf{T})_X=\mathbf{T}_{(R,\,\Theta-\Theta_0,\,Z-Z_0)}$. Applying this to $\mathbf{C}^\flat$, one obtains $(g\cdot \mathbf{C}^\flat)_X=\mathbf{C}^\flat_{(R,\,\Theta-\Theta_0,\,Z-Z_0)}$. Since $\mathbf{C}^\flat$ depends only on $R$, one has $\mathbf{C}^\flat_{(R,\,\Theta-\Theta_0,\,Z-Z_0)}=\mathbf{C}^\flat_X$, and hence
%-----------------------------
\begin{equation}
	(g\cdot \mathbf{C}^\flat)_X=\mathbf{C}^\flat_X\,,
\end{equation}
%-----------------------------
which shows that $\mathbf{C}^\flat$ is invariant under the action of $\mathrm{SO}(2)\times \mathrm{T}(1)\subset \mathrm{SE}(3)$.

The tensor $\mathbf{b}^{\sharp}$ has the representation
%-----------------------------
\begin{equation} 
	[b^{ab}]=\begin{bmatrix}
	 \frac{R^2}{K \left(C_5 K+R^2\right)} & 0 & 0 \\
	 0 & \frac{C_1^2}{R^2}+C_2^2 & \frac{C_1 C_3}{R^2}+C_2  C_4 \\
	 0 & \frac{C_1 C_3}{R^2}+C_2 C_4 &  \frac{C_3^2}{R^2}+C_4^2
	\end{bmatrix}
	\,,
\end{equation}
%-----------------------------
which, when expressed in spatial coordinates, depends only on $r$. Consider the action of $\mathrm{SO}(2)\times \mathrm{T}(1)$ on the spatial configuration,
%-----------------------------
\begin{equation}
	\theta\mapsto \bar{\theta}=\theta+\theta_0\,,\qquad z\mapsto \bar{z}=z+z_0\,.
\end{equation}
%-----------------------------
For this action, the inverse map is $\theta\mapsto \theta-\theta_0$, $z\mapsto z-z_0$, and the tangent map preserves the orthonormal frame $\{\mathbf{e}_r,\mathbf{e}_\theta,\mathbf{e}_z\}$. Hence, the prolonged action reduces to $(g\cdot\mathbf{b}^{\sharp})_{\mathbf{x}}=\mathbf{b}^{\sharp}_{(r,\,\theta-\theta_0,\,z-z_0)}$. Since $\mathbf{b}^{\sharp}$ depends only on $r$, one has $\mathbf{b}^{\sharp}_{(r,\,\theta-\theta_0,\,z-z_0)}=\mathbf{b}^{\sharp}_{\mathbf{x}}$, and therefore
%-----------------------------
\begin{equation}
	(g\cdot \mathbf{b}^{\sharp})_{\mathbf{x}}=\mathbf{b}^{\sharp}_{\mathbf{x}}\,,
\end{equation}
%-----------------------------
which shows that $\mathbf{b}^{\sharp}$ is invariant under the action of $\mathrm{SO}(2)\times \mathrm{T}(1)\subset \mathrm{SE}(3)$.
%-----------------------------

%-----------------------------
%-----------------------------
\paragraph{Family 4: Inflation/inversion of a sector of a spherical shell.}
In spherical coordinates in the reference and current configurations,
%-----------------------------
\begin{equation}\label{eqn:Family4-u}
	r^3=\pm R^3 + C_1^3\,, \qquad
	\theta=\pm \Theta\,, \qquad
	\phi=\Phi\,.
\end{equation}
%-----------------------------
The deformation gradient is diagonal in the spherical basis, and the associated tensor $\mathbf{C}^{\flat}$ has the representation
%-----------------------------
\begin{equation} 
	[C_{AB}]=\begin{bmatrix}
	\frac{R^4}{\left(C_1^3 \pm R^3\right)^{4/3}} & 0 & 0\\
	0 & \left(C_1^3 \pm R^3\right)^{2/3} & 0 \\
	0 & 0 & \left(C_1^3 \pm R^3\right)^{2/3} \sin^2\Theta
	\end{bmatrix}
	\,,
\end{equation}
%-----------------------------
which depends only on $R$. Equivalently, one may write
%-----------------------------
\begin{equation} 
	\mathbf{C}^{\flat}(\mathbf{X})
	=\frac{R^4}{\left(C_1^3 \pm R^3\right)^{4/3}}\hat{\mathbf{R}}\otimes\hat{\mathbf{R}}
	+\frac{\left(C_1^3 \pm R^3\right)^{2/3}}{R^2}(\mathbf{I}-\hat{\mathbf{R}}\otimes\hat{\mathbf{R}})
	\,,
\end{equation}
%-----------------------------
where $\hat{\mathbf{R}}=\frac{\mathbf{X}}{|\mathbf{X}|}$. 

Consider the action of $\mathrm{SO}(3)$ on $\mathcal{B}$,
%-----------------------------
\begin{equation}
	\mathbf{X}\mapsto \bar{\mathbf{X}}=\mathbf{Q}\mathbf{X}\,, 
	\qquad \forall\,\mathbf{Q}\in \mathrm{SO}(3)\,.
\end{equation}
%-----------------------------
For this action, the inverse map is $\mathbf{X}\mapsto \mathbf{Q}^{\mathsf{T}}\mathbf{X}$.
Moreover, the map is linear, so its tangent map equals the same linear transformation, i.e., $T_X\rho_g=\mathbf{Q}$, while the tangent map of the inverse is $T_X\rho_{g^{-1}}=\mathbf{Q}^{\mathsf{T}}$.
Hence, the prolonged action on a $(0,2)$-tensor field $\mathbf{T}$ reduces to
$(g\cdot\mathbf{T})_X(\mathbf{U},\mathbf{V})=\mathbf{T}_{\mathbf{Q}^{\mathsf{T}}\mathbf{X}}(\mathbf{Q}^{\mathsf{T}}\mathbf{U},\mathbf{Q}^{\mathsf{T}}\mathbf{V})$.  Applying this to $\mathbf{C}^\flat$, one obtains
$(g\cdot\mathbf{C}^\flat)_X(\mathbf{U},\mathbf{V})=\mathbf{C}^\flat_{\mathbf{Q}^{\mathsf{T}}\mathbf{X}}(\mathbf{Q}^{\mathsf{T}}\mathbf{U},\mathbf{Q}^{\mathsf{T}}\mathbf{V})$.

Let $R=|\mathbf{X}|$ and $\hat{\mathbf{R}}=\mathbf{X}/|\mathbf{X}|$. Then $|\mathbf{Q}^{\mathsf{T}}\mathbf{X}|=|\mathbf{X}|=R$ and $\widehat{\mathbf{Q}^{\mathsf{T}}\mathbf{X}}=\mathbf{Q}^{\mathsf{T}}\hat{\mathbf{R}}$.
Thus, $\hat{\mathbf{R}}\otimes\hat{\mathbf{R}} \mapsto \mathbf{Q}^{\mathsf{T}}(\hat{\mathbf{R}}\otimes\hat{\mathbf{R}})\mathbf{Q}$ and $\mathbf{I}-\hat{\mathbf{R}}\otimes\hat{\mathbf{R}} \mapsto \mathbf{Q}^{\mathsf{T}}(\mathbf{I}-\hat{\mathbf{R}}\otimes\hat{\mathbf{R}})\mathbf{Q}$. Therefore, if
$\mathbf{C}^\flat = A(R)\,\hat{\mathbf{R}}\otimes\hat{\mathbf{R}} + B(R)\,(\mathbf{I}-\hat{\mathbf{R}}\otimes\hat{\mathbf{R}})$, then $\mathbf{C}^\flat_{\mathbf{Q}^{\mathsf{T}}\mathbf{X}} = \mathbf{Q}^{\mathsf{T}}\mathbf{C}^\flat_{\mathbf{X}}\mathbf{Q}$. Substituting into the prolonged action gives us
%-----------------------------
\begin{align}
	(g\cdot\mathbf{C}^\flat)_{\mathbf{X}}(\mathbf{U},\mathbf{V})
	&=\mathbf{C}^\flat_{\mathbf{Q}^{\mathsf{T}}\mathbf{X}}
	(\mathbf{Q}^{\mathsf{T}}\mathbf{U},\mathbf{Q}^{\mathsf{T}}\mathbf{V}) \nonumber\\
	&=\llangle 
	\mathbf{Q}^{\mathsf{T}}\mathbf{U}\,,\,
	\mathbf{C}^{\sharp}_{\mathbf{Q}^{\mathsf{T}}\mathbf{X}}
	\mathbf{Q}^{\mathsf{T}}\mathbf{V}
	\rrangle_{\mathbf{G}} \nonumber\\
	&=\llangle 
	\mathbf{Q}^{\mathsf{T}}\mathbf{U}\,,\,
	\mathbf{Q}^{\mathsf{T}}
	\mathbf{C}^{\sharp}_{\mathbf{X}}
	\mathbf{Q}\mathbf{Q}^{\mathsf{T}}\mathbf{V}
	\rrangle_{\mathbf{G}} \nonumber\\
	&=\llangle 
	\mathbf{Q}^{\mathsf{T}}\mathbf{U}\,,\,
	\mathbf{Q}^{\mathsf{T}}
	\mathbf{C}^{\sharp}_{\mathbf{X}}
	\mathbf{V}
	\rrangle_{\mathbf{G}} \nonumber\\
	&=\llangle 
	\mathbf{U}\,,\,
	\mathbf{C}^{\sharp}_{\mathbf{X}}
	\mathbf{V}
	\rrangle_{\mathbf{G}} \nonumber\\
	&=\mathbf{C}^\flat_{\mathbf{X}}(\mathbf{U},\mathbf{V})\,,
\end{align}
%-----------------------------
where orthogonality of $\mathbf{Q}$ was used.
Hence,
%-----------------------------
\begin{equation}
	(g\cdot \mathbf{C}^\flat)_{\mathbf{X}}=\mathbf{C}^\flat_{\mathbf{X}}\,,
\end{equation}
%-----------------------------
which shows that $\mathbf{C}^\flat$ is invariant under the action of $\mathrm{SO}(3)\subset \mathrm{SE}(3)$.

The tensor $\mathbf{b}^{\sharp}$ has the representation
%-----------------------------
\begin{equation} 
	[b^{ab}]=\begin{bmatrix}
	 C_1^2 & \frac{C_1 C_2}{R} & 0 \\
	 \frac{C_1 C_2}{R} & \frac{C_2^2+C_3^2}{R^2} & 0 \\
	 0 & 0 & \frac{1}{C_1^4 C_3^2}   
	\end{bmatrix}
	\,,
\end{equation}
%-----------------------------
which, when expressed in spatial coordinates, depends only on $r$. Consider the action of $\mathrm{SO}(3)$ on the spatial configuration,
%-----------------------------
\begin{equation}
	\mathbf{x}\mapsto \bar{\mathbf{x}}=\mathbf{Q}\mathbf{x}\,, \qquad \forall\,\mathbf{Q}\in \mathrm{SO}(3)\,.
\end{equation}
%-----------------------------
For this action, the inverse map is $\mathbf{x}\mapsto \mathbf{Q}^{\mathsf{T}}\mathbf{x}$, and the tangent map is $T_{\mathbf{x}}\rho_{g^{-1}}=\mathbf{Q}^{\mathsf{T}}$. Hence, the prolonged action on a $(2,0)$-tensor field reduces to
$(g\cdot\mathbf{b}^{\sharp})_{\mathbf{x}}=\mathbf{Q}\mathbf{b}^{\sharp}_{\mathbf{Q}^{\mathsf{T}}\mathbf{x}}\mathbf{Q}^{\mathsf{T}}$. Since $\mathbf{b}^{\sharp}$ depends only on $r$, and since $|\mathbf{Q}^{\mathsf{T}}\mathbf{x}|=|\mathbf{x}|$, one has $\mathbf{b}^{\sharp}_{\mathbf{Q}^{\mathsf{T}}\mathbf{x}}=\mathbf{b}^{\sharp}_{\mathbf{x}}$ in the rotated spherical frame. Therefore,
%-----------------------------
\begin{equation}
	(g\cdot \mathbf{b}^{\sharp})_{\mathbf{x}}=\mathbf{b}^{\sharp}_{\mathbf{x}}\,,
\end{equation}
%-----------------------------
which shows that $\mathbf{b}^{\sharp}$ is invariant under the action of $\mathrm{SO}(3)\subset \mathrm{SE}(3)$.
%-----------------------------

%-----------------------------
%-----------------------------
\paragraph{Family 5: Inflation, bending, extension, and azimuthal shearing of an annular wedge.}
In cylindrical coordinates in the reference and current configurations,
%-----------------------------
\begin{equation}\label{eqn:Family5-u}
	r=C_1 R\,, \qquad
	\theta=C_2\log R + C_3\Theta + C_4\,, \qquad
	z=\frac{1}{C_1^2C_3}Z+C_5\,.
\end{equation}
%-----------------------------
The parameters $C_4$ and $C_5$ correspond to rigid motions and can be disregarded. The deformation gradient reads
%-----------------------------
\begin{equation}\label{eqn:Family5-F}
	\mathbf{F}
	=
	C_1\,\mathbf{e}_r\otimes \mathbf{E}_R
	+\frac{C_2}{R}\,\mathbf{e}_\theta\otimes \mathbf{E}_R
	+C_3\,\mathbf{e}_\theta\otimes \mathbf{E}_\Theta
	+\frac{1}{C_1^2C_3}\,\mathbf{e}_z\otimes \mathbf{E}_Z\,,
\end{equation}
%-----------------------------
and in coordinates
%-----------------------------
\begin{equation}
	\left[F^a{}_A\right]
	=
	\begin{bmatrix}
	C_1 & 0 & 0\\
	\frac{C_2}{R} & C_3 & 0\\
	0 & 0 & \frac{1}{C_1^2C_3}
	\end{bmatrix}.
\end{equation}
%-----------------------------
The associated tensor $\mathbf{C}^{\flat}$ has the representation
%-----------------------------
\begin{equation} 
	[C_{AB}]=\begin{bmatrix}
	C_1^2 \left(C_2^2+1\right) & C_1^2 C_2 C_3 R & 0 \\
        C_1^2 C_2 C_3 R & C_1^2 C_3^2 R^2 & 0 \\
	0 & 0 & \frac{1}{C_1^4 C_3^2}
	\end{bmatrix}
	\,,
\end{equation}
%-----------------------------
which depends only on $R$. 
Consider the action of $\mathrm{SO}(2)\times \mathrm{T}(1)$ on $\mathcal{B}$,
%-----------------------------
\begin{equation}
	\Theta\mapsto \bar{\Theta}=\Theta+\Theta_0\,, \qquad
	Z\mapsto \bar{Z}=Z+Z_0\,.
\end{equation}
%-----------------------------
For this action, the inverse map is $\Theta\mapsto \Theta-\Theta_0$, $Z\mapsto Z-Z_0$, and the tangent map preserves the orthonormal frame $\{\mathbf{E}_R,\mathbf{E}_\Theta,\mathbf{E}_Z\}$. Hence, the prolonged action reduces to $(g\cdot\mathbf{T})_X=\mathbf{T}_{(R,\,\Theta-\Theta_0,\,Z-Z_0)}$. Applying this to $\mathbf{C}^\flat$, one obtains $(g\cdot \mathbf{C}^\flat)_X=\mathbf{C}^\flat_{(R,\,\Theta-\Theta_0,\,Z-Z_0)}$.
Since $\mathbf{C}^\flat$ depends only on $R$, one has $\mathbf{C}^\flat_{(R,\,\Theta-\Theta_0,\,Z-Z_0)}=\mathbf{C}^\flat_X$, and hence
%-----------------------------
\begin{equation}
	(g\cdot \mathbf{C}^\flat)_X=\mathbf{C}^\flat_X\,,
\end{equation}
%-----------------------------
which shows that $\mathbf{C}^\flat$ is invariant under the action of $\mathrm{SO}(2)\times \mathrm{T}(1)\subset \mathrm{SE}(3)$.

The tensor $\mathbf{b}^{\sharp}$ has the representation
%-----------------------------
\begin{equation}
	[b^{ab}]=\begin{bmatrix}
	 C_1^2 & \frac{C_1 C_2}{R} & 0 \\
	 \frac{C_1 C_2}{R} & \frac{C_2^2+C_3^2}{R^2} & 0 \\
	 0 & 0 & \frac{1}{C_1^4 C_3^2}
	\end{bmatrix}
	\,,
\end{equation}
%-----------------------------
which, when expressed in spatial coordinates, depends only on $r$. Consider the action of $\mathrm{SO}(2)\times \mathrm{T}(1)$ on the spatial configuration,
%-----------------------------
\begin{equation}
	\theta\mapsto \bar{\theta}=\theta+\theta_0\,, \qquad z\mapsto \bar{z}=z+z_0\,.
\end{equation}
%-----------------------------
For this action, the inverse map is $\theta\mapsto \theta-\theta_0$, $z\mapsto z-z_0$, and the tangent map preserves the orthonormal frame $\{\mathbf{e}_r,\mathbf{e}_\theta,\mathbf{e}_z\}$. Hence, the prolonged action reduces to $(g\cdot\mathbf{b}^{\sharp})_{\mathbf{x}}=\mathbf{b}^{\sharp}_{(r,\,\theta-\theta_0,\,z-z_0)}$. Since $\mathbf{b}^{\sharp}$ depends only on $r$, one has $\mathbf{b}^{\sharp}_{(r,\,\theta-\theta_0,\,z-z_0)}=\mathbf{b}^{\sharp}_{\mathbf{x}}$, and therefore
%-----------------------------
\begin{equation}
	(g\cdot \mathbf{b}^{\sharp})_{\mathbf{x}}=\mathbf{b}^{\sharp}_{\mathbf{x}}\,,
\end{equation}
%-----------------------------
which shows that $\mathbf{b}^{\sharp}$ is invariant under the action of $\mathrm{SO}(2)\times \mathrm{T}(1)\subset \mathrm{SE}(3)$.
%-----------------------------

%-----------------------------
%-----------------------------
\subsection{Symmetry structure of the universal deformations}

For each of the above families, the right Cauchy--Green tensor $\mathbf{C}^\flat$ is invariant under the prolonged action of a Lie subgroup of $\mathrm{SE}(3)$ acting on the reference configuration. Similarly, the left Cauchy--Green tensor $\mathbf{b}^\sharp$ admits a symmetry group acting on the current configuration. These symmetry groups need not coincide, although their dimensions agree, as they are related through the deformation.
The corresponding symmetry groups are summarized in Table \ref{table:symmetry groups}. They are expressed in terms of the rotation group $\mathrm{SO}(n)$ and the translation group $\mathrm{T}(n)$.

%-----------------------------
%-----------------------------
{\renewcommand{\arraystretch}{1}
\begin{table}[ht!]
\begin{center}
\begin{tabular}{|c|c|c|c|c|}
\hline 
Family & $\mathbf{C}^\flat$ & Dimension & $\mathbf{b}^\sharp$ & Dimension\\ \hline
0&$\text{T}(3)$&3&$\text{T}(3)$&3\\
1&$\text{T}(2)$&2&$\text{SO}(2)\times\text{T}(1)$&2\\
2&$\text{SO}(2)\times\text{T}(1)$&2&$\text{T}(2)$&2\\
3&$\text{SO}(2)\times\text{T}(1)$&2&$\text{SO}(2)\times\text{T}(1)$&2\\
4&$\text{SO}(3)$&3&$\text{SO}(3)$&3\\
5&$\text{SO}(2)\times\text{T}(1)$&2&$\text{SO}(2)\times\text{T}(1)$&2\\ \hline
\end{tabular}
\end{center}
\caption[]{Lie subgroups of $\text{SE}(3)$ under whose prolonged action the universal stretch tensor fields remain invariant. The symmetry groups are computed explicitly for both $\mathbf{C}^{\flat}$ and $\mathbf{b}^{\sharp}$; note in particular that for Family $4$ the invariance group is $\mathrm{SO}(2)$ rather than $\mathrm{SO}(3)$ \citep{Goodbrake2020}.}
\label{table:symmetry groups}
\end{table}
}
%-----------------------------
%-----------------------------

An important observation is that the dimension of the symmetry group influences the structure of the equilibrium equations. Families with three-dimensional symmetry groups satisfy equilibrium identically, whereas those with two-dimensional symmetry require additional constraints.
In summary, each universal deformation is characterized by the invariance of $\mathbf{C}^\flat$ under the prolonged action of a Lie subgroup of $\mathrm{SE}(3)$ acting on the reference configuration. This invariance will be used in the symmetry reduction of the system of nonlinear universality PDEs.
More specifically, we assume that the residual stress field $\mathring{\mathbf{S}}$ has the same symmetry as $\mathbf{C}^\flat$.

%-----------------------------
%-----------------------------
\section{Universal Residual Stresses for the Six Known Families of Universal Deformations} \label{Universal-R-Stressed}

In this section, for each of the known families of universal deformations, we find the corresponding universal residual stress fields assuming that the residual stress field has the same symmetry as the universal deformation.

%-----------------------------
%-----------------------------
\subsection{Family 0: Homogeneous deformations}

For homogeneous deformations, the deformation mapping is given by $x^a(\mathbf{X}) = F^a{}_A X^A + c^a$, where $[F^a{}_A]$ is a constant matrix and $c^a$ are the components of a constant vector. In Cartesian coordinates, the incompressibility constraint is written as $\det [F^a{}_A] = 1$.  
The right Cauchy–Green strain tensor in Cartesian coordinates has components $C_{AB} = F^a{}_A F^b{}_B \,\delta_{ab}$, which are constant.  
The tensors $[b^{ab}]$ and $[c^{ab}]$ are constant matrices, and both $I_1$ and $I_2$ are constant.  

For a homogeneous deformation, $\mathbf{b}$ is a constant tensor and its eigenvalues and eigenvectors are constant in space. The universality constraints then imply that, whenever $I_4$, $I_5$, and $I_6$ are not constant, their gradients $\nabla I_4$, $\nabla I_5$, and $\nabla I_6$ are all parallel to a single fixed eigenvector $\boldsymbol{\mathsf{e}}$ of $\mathbf{b}$. 
Let us choose spatial Cartesian coordinates $\{x^1,x^2,x^3\}$ such that $\boldsymbol{\mathsf{e}}=\mathbf{e}_1=\frac{\partial}{\partial x^1}$. Consequently, the invariants can vary only along $x^1$-direction and are locally of the form $I_j(x^1)$. 
In this case,
%-----------------------------
\begin{equation}
	\nabla I_j = I_j'(x^1)\,\mathbf{e}_1\,,\qquad
	\nabla\nabla I_j = I_j''(x^1)\,\mathbf{e}_1\otimes \mathbf{e}_1\,,
\end{equation}
%-----------------------------
and therefore the Hessian has rank at most one, annihilates any direction orthogonal to $\mathbf{e}_1$, and has $\mathbf{e}_1$ as an eigenvector with eigenvalue $I_j''(x^1)$.
In components, we have $I_{j,a}=I'_j\,\delta^1_a$ and $I_{j|ma}=I''_j\,\delta^1_a\,\delta^1_m$.

Recall that the Finger tensor has the following spectral representation:
%---------------------------------
\begin{equation}
	\mathbf{b}^\sharp= 
	\lambda_1^2 \,\n \otimes \n+\lambda_2^2 \,\nn \otimes\nn
	+\lambda_3^2 \,\nnn \otimes\nnn\,,
\end{equation}
%---------------------------------
where $\lambda_i$ are the principal stretches and $\n$, $\nn$, and $\nn$ are eigenvalues of the Finger tensor. When $\n=\mathbf{e}_1$, $b^{12}=b^{13}=0$, and hence, $b^1_2=b^2_1=b^1_3=b^3_1=0$.

Recall that $\mathring{\boldsymbol{\sigma}}=\mathbf{F}\mathring{\mathbf{S}}\mathbf{F}^\star$, or equivalently, $\mathring{\mathbf{S}}=\mathbf{F}^{-1}\,\mathring{\boldsymbol{\sigma}}\,\mathbf{F}^{-\star}$. Thus,\footnote{In components, %---------------------------------
\begin{equation}
	0=\mathring{S}^{AB}{}_{|B}
	=F^{-A}{}_a\,\mathring{\sigma}^{ab}{}_{|B}\,F^{-B}{}_b
	=F^{-A}{}_a\,\mathring{\sigma}^{ab}{}_{|c}\,F^c{}_B\,F^{-B}{}_b
	=F^{-A}{}_a\,\mathring{\sigma}^{ab}{}_{|c}\,\delta^c_b
	=F^{-A}{}_a\,\mathring{\sigma}^{ab}{}_{|b} \,.
\end{equation}
%---------------------------------
}
%---------------------------------
\begin{equation}
	\mathbf{0}=\operatorname{Div}\mathring{\mathbf{S}}
	=\mathbf{F}^{-1}\,(\operatorname{Div}
	\mathring{\boldsymbol{\sigma}})\,\mathbf{F}^{-\star}
	=\mathbf{F}^{-1}\,(\operatorname{div}
	\mathring{\boldsymbol{\sigma}})\mathbf{F}\,\mathbf{F}^{-\star}
	=\mathbf{F}^{-1}\,\operatorname{div} \mathring{\boldsymbol{\sigma}}\,.
\end{equation}
%---------------------------------
Therefore, for homogeneous deformations $\operatorname{div}\mathring{\boldsymbol{\sigma}}=\mathbf{0}$.
In this case, if principal invariants of $\mathring{\boldsymbol{\sigma}}$ are constant, $\mathring{\boldsymbol{\sigma}}$ would be a homogeneous residual stress, and hence, must vanish in order to satisfy the traction-free boundary conditions.

Note that $\mathcal{B}^{(j)}_{ab} = I''_j\,\delta^1_b\,\delta^1_m\, b_a^m= I''_j\,\delta^1_b\, b_a^1$. The corresponding universality constraints are $I''_j (\delta^1_b\, b_a^1-\delta^1_a\, b_b^1)=0$, which are trivially satisfied. Similarly, the universality constraints corresponding to $\mathcal{C}^{(j)}_{ab}$ are trivially satisfied as well. 

For homogeneous deformations the non-trivial universality constraints are symmetry of the following terms in $(a,b)$ where $j=4,\cdots, 10$ ($8$ terms) and $j\leq k =4,\cdots, 10$:
%-----------------------------
\begin{align}
	\label{Zj-Universal}
	\mathcal{Z}^{(j)}_{ab} &	 = I_{j|mb}\, \mathring{\sigma}_a^m 
	+ I_{j,m}\, \mathring{\sigma}_a^m{}_{|b} \,,\\
	\mathcal{W}_{ab} & = \mathring{\sigma}_a^n{}_{|mb}\, \mathring{\sigma}_n^m
	+ \mathring{\sigma}_a^n{}_{|m}\, \mathring{\sigma}_n^m{}_{|b}\,, \\
	\mathcal{W}^{(j)}_{ab} & = I_{j|mb}\, \mathring{\sigma}_a^n\, \mathring{\sigma}_n^m 
	+ I_{j,m}\,(\mathring{\sigma}_a^n{}_{|b}\, \mathring{\sigma}_n^m
             + \mathring{\sigma}_a^n\, \mathring{\sigma}_n^m{}_{|b}) 
	+ I_{j,b}\,\mathring{\sigma}_a^n{}_{|m}\, \mathring{\sigma}_n^m  \,,\\
	\mathcal{X}_{ab} & = \mathring{\sigma}_a^n{}_{|mb}\, b_n^m	\,, \\
	\mathcal{X}^{(j)}_{ab} & = I_{j|mb}\,(b_a^n\, \mathring{\sigma}_n^m 
	+ \mathring{\sigma}_a^n\, b_n^m) 
	+ I_{j,m}\,(b_a^n\, \mathring{\sigma}_n^m{}_{|b} 
         + \mathring{\sigma}_a^n{}_{|b}\, b_n^m ) 
         + I_{j,b}\,\mathring{\sigma}_a^n{}_{|m}\, b_n^m \,,\\
	\mathcal{H}_{ab} & =  \mathring{\sigma}_a^n{}_{|mb}\, c_n^m\,, \\
	\mathcal{H}^{(j)}_{ab} & = I_{j|mb} \left(c_a^n\, \mathring{\sigma}_n^m 
	+ \mathring{\sigma}_a^n\, c_n^m \right) 
	+ I_{j,m} \left(c_a^n\, \mathring{\sigma}_n^m{}_{|b} 
         + \mathring{\sigma}_a^n{}_{|b}\, c_n^m  \right)  
         + I_{j,b} \,\mathring{\sigma}_a^n{}_{|m}\, c_n^m  \,,\\
	\mathcal{K}_{ab} & = b_a^n\,\mathring{\sigma}^{2m}_{~n|mb} 
        + \mathring{\sigma}^{2n}_{~a|mb}\,b_n^m\,, \\
	\mathcal{K}^{(j)}_{ab} & = I_{j|bm}\,\big(b_a^n\,\mathring{\sigma}^{2m}_n 
	+ \mathring{\sigma}^{2n}_a\,b_n^m\big) 
	+ I_{j,m}\,\big(b_a^n\,\mathring{\sigma}^{2m}_{~n|b}
         + \mathring{\sigma}^{2n}_{~a|b}\,b_n^m \big) 
	+ I_{j,b}\,\big(b_a^n\,\mathring{\sigma}^{2m}_{~n|m}
         + \mathring{\sigma}^{2n}_{~a|m}\,b_n^m \big)  \,.
\end{align}
%-----------------------------
Knowing that the ambient space is flat, second-order covariant derivatives are symmetric and hence $\mathcal{X}_{ab} = \mathring{\sigma}_a^n{}_{|bm}\, b_n^m$. The corresponding universality constraints are: 
%---------------------------------
\begin{equation}
	(\mathring{\sigma}_a^n{}_{|bm}-\mathring{\sigma}_b^n{}_{|am})\, 
	b_n^m=(\mathring{\sigma}_a^n{}_{|b}-\mathring{\sigma}_b^n{}_{|a})_{|m}\, b_n^m=0\,.
\end{equation}
%---------------------------------
Similarly, symmetry of $\mathcal{H}_{ab}$ implies that 
%---------------------------------
\begin{equation}
	(\mathring{\sigma}_a^n{}_{|b}-\mathring{\sigma}_b^n{}_{|a})_{|m}\, c_n^m=0\,.
\end{equation}
%---------------------------------

Note that $\mathcal{W}_{ab}=\mathring{\sigma}_a^n{}_{|mb}\,\mathring{\sigma}_n^m
+\mathring{\sigma}_a^n{}_{|m}\,\mathring{\sigma}_n^m{}_{|b}
=\left(\mathring{\sigma}_a^n{}_{|m}\,\mathring{\sigma}_n^m\right)_{|b}$.
Therefore, symmetry of $\mathcal{W}_{ab}$ in $(a,b)$ implies that
%---------------------------------
\begin{equation}
	\left(\mathring{\sigma}_a^n{}_{|m}\,\mathring{\sigma}_n^m\right)_{|b}
	=
	\left(\mathring{\sigma}_b^n{}_{|m}\,\mathring{\sigma}_n^m\right)_{|a}\,.
\end{equation}
%---------------------------------
Hence, there exists a scalar field $\phi$ such that
%---------------------------------
\begin{equation}
	\mathring{\sigma}_a^n{}_{|m}\,\mathring{\sigma}_n^m=\phi_{,a}\,.
\end{equation}
%---------------------------------
In other words, $\operatorname{div}\mathring{\boldsymbol{\sigma}}^2$ is a gradient.

The universality constraint \eqref{Zj-Universal} gives us the following for $j=4,5,6$:
%---------------------------------
\begin{equation}\label{Ij-Derivatives}
\begin{dcases}
	\left(\mathring{\sigma}_1^1{}_{|2}-\mathring{\sigma}_2^1{}_{|1}\right)I'_j
	-\mathring{\sigma}_1^2 I''_j=0\,,\\
	\left(\mathring{\sigma}_1^1{}_{|3}-\mathring{\sigma}_3^1{}_{|1}\right)I'_j
	-\mathring{\sigma}_1^3 I''_j=0\,,\\
	\left(\mathring{\sigma}_2^1{}_{|3}-\mathring{\sigma}_3^1{}_{|2}\right)I'_j=0
	\,.
\end{dcases}
\end{equation}
%---------------------------------
If $I'_j = 0$, then $I_j$ is constant and \eqref{Ij-Derivatives} is satisfied identically.
Assume now that $I'_j \neq 0$. We have the following possibilities:

\noindent \textbf{(i)} If $\mathring{\sigma}_1^2 \neq 0$ and $\mathring{\sigma}_1^3 = 0$, then the first equation gives
%---------------------------------
\begin{equation}
	f'(x^1)= \frac{\mathring{\sigma}_1^1{}_{|2}-\mathring{\sigma}_2^1{}_{|1}}{\mathring{\sigma}_1^2}\,,
	\qquad I'_j(x^1)=c_j\,e^{f(x^1)}\,.
\end{equation}
%---------------------------------

\noindent \textbf{(ii)} If $\mathring{\sigma}_1^2 = 0$ and $\mathring{\sigma}_1^3 \neq 0$, then the second equation gives
%---------------------------------
\begin{equation}
	f'(x^1)= \frac{\mathring{\sigma}_1^1{}_{|3}-\mathring{\sigma}_3^1{}_{|1}}{\mathring{\sigma}_1^3}\,,
	\qquad I'_j(x^1)=c_j\,e^{f(x^1)}\,.
\end{equation}
%---------------------------------

\noindent \textbf{(iii)} If $\mathring{\sigma}_1^2 \neq 0$ and $\mathring{\sigma}_1^3 \neq 0$, then both equations must hold and hence
%---------------------------------
\begin{equation}
	f'(x^1) =\frac{\mathring{\sigma}_1^1{}_{|2}-\mathring{\sigma}_2^1{}_{|1}}{\mathring{\sigma}_1^2}
	= \frac{\mathring{\sigma}_1^1{}_{|3}-\mathring{\sigma}_3^1{}_{|1}}{\mathring{\sigma}_1^3}\,,
\end{equation}
%---------------------------------
which is a differential constraint on $\mathring{\boldsymbol{\sigma}}$, and in this case
%---------------------------------
\begin{equation}
	I'_j(x^1)=c_j\,e^{f(x^1)}\,.
\end{equation}
%---------------------------------

\noindent \textbf{(iv)} If $\mathring{\sigma}_1^2 = \mathring{\sigma}_1^3 = 0$, the first two equations reduce to differential restrictions on $\mathring{\boldsymbol{\sigma}}$.

We observe that even for the simplest family, namely homogeneous deformations, solving the full system of universality PDEs is not tractable. Assuming that $\mathring{\mathbf{S}}$ has the same symmetry as $\mathbf{C}^\flat$ does, $\mathring{\mathbf{S}}$ must be homogeneous. For such residual stress fields all the universality constraints are trivially satisfied.

%-----------------------------
%-----------------------------
\subsection{Family 1: Bending, stretching, and shearing of a rectangular block}

For this family of deformations given in \eqref{eqn:Family1-u}, we use the Cartesian $(X,Y,Z)$ and cylindrical $(r,\theta,z)$ coordinates in the reference and current configurations, respectively. Assuming that the residual stress field has the same symmetry as $\mathbf{C}^\flat$, we start with the following ansatz 
%-----------------------------
\begin{equation} 
	[\mathring{S}^{AB}]=\begin{bmatrix}
	f_1(X) & f_4(X) & f_6(X)\\
	f_4(X) & f_2(X) & f_5(X) \\
	f_6(X) & f_5(X) & f_3(X)
	\end{bmatrix}
	\,,
\end{equation}
%-----------------------------
which includes six unknown functions $f_i$, $i=1,\hdots 6$, of the single variable $X$.

The universality PDE corresponding to $\accentset{s}{\mathcal{Z}}^{(12)}_{23}= \accentset{s}{\mathcal{Z}}^{(12)}_{32}$ implies that $f_4(X)=0$. Similarly, the universality PDE corresponding to $\accentset{s}{\mathcal{Z}}^{(12)}_{13}= \accentset{s}{\mathcal{Z}}^{(12)}_{31}$ implies that $f_6(X)=0$. The remaining universality constraints are identically satisfied for arbitrary $f_1(X)$, $f_2(X)$, $f_3(X)$, and $f_5(X)$.\footnote{The symbolic manipulation and simplification of the universality PDEs for the universal deformation families and the associated residual stress fields were performed using \textit{Wolfram Mathematica} Version 14.3.}
Therefore, the universal residual stress fields have the following form
%-----------------------------
\begin{equation} 
	[\mathring{S}^{AB}]
	=
	\begin{bmatrix}
	A(X) & 0 & 0\\
	0 & B(X) & D(X) \\
	0 & D(X) & C(X)
	\end{bmatrix}
	\,,
\end{equation}
%-----------------------------
where $A$, $B$, $C$, and $D$ are arbitrary twice-differentiable functions of $X$.

%-----------------------------
%-----------------------------
\subsection{Family 2: Straightening, stretching, and shearing of a sector of a cylindrical shell}

This family of deformations, with respect to the cylindrical $(R,\Theta,Z)$ and Cartesian $(x,y,z)$ coordinates in the reference and current configurations, respectively, have the representation \eqref{eqn:Family2-u}.
Motivated by the structure of $\mathbf{C}^\flat$ in \eqref{C-Family2}, we seek universal residual stress fields of the form
%-----------------------------
\begin{equation} 
	[\mathring{S}^{AB}]
	=
	\begin{bmatrix}
	f_1(R) & f_4(R) & f_6(R)\\
	f_4(R) & f_2(R) & f_5(R) \\
	f_6(R) & f_5(R) & f_3(R)
	\end{bmatrix}
	\,,
\end{equation}
%-----------------------------
where $f_i$, $i=1,\ldots,6$, are unknown functions of $R$. Substituting this ansatz into the universality PDEs, the condition $\accentset{s}{\mathcal{Z}}^{(12)}_{12}= \accentset{s}{\mathcal{Z}}^{(12)}_{21}$ yields $f_4(R)=0$, while the condition $\accentset{s}{\mathcal{Z}}^{(11)}_{13}= \accentset{s}{\mathcal{Z}}^{(11)}_{31}$ yields $f_6(R)=0$. All remaining universality constraints are then satisfied identically without imposing any further restrictions on $f_1(R)$, $f_2(R)$, $f_3(R)$, and $f_5(R)$. Consequently, the universal residual stress fields are given by
%-----------------------------
\begin{equation} 
	[\mathring{S}^{AB}]
	=
	\begin{bmatrix}
	A(R) & 0 & 0\\
	0 & B(R) & D(R) \\
	0 & D(R) & C(R)
	\end{bmatrix}
	\,,
\end{equation}
%-----------------------------
where $A$, $B$, $C$, and $D$ are arbitrary twice-differentiable functions of $R$.

%-----------------------------
%-----------------------------
\subsection{Family 3: Inflation, bending, torsion, extension, and shearing of a sector of an annular wedge}

With respect to the cylindrical coordinates $(R,\Theta,Z)$ and $(r,\theta,z)$ in the reference and current configurations, respectively, this family of deformations have the representation given in \eqref{eqn:Family3-u}.
We begin with the most general symmetric residual stress tensor consistent with the radial symmetry of the deformation, namely,
%-----------------------------
\begin{equation} 
	[\mathring{S}^{AB}]
	= \begin{bmatrix}
	f_1(R) & f_4(R) & f_6(R)\\
	f_4(R) & f_2(R) & f_5(R) \\
	f_6(R) & f_5(R) & f_3(R)
	\end{bmatrix}
	\,,
\end{equation}
%-----------------------------
where the six functions $f_i(R)$ are to be determined from the universality constraints. 
Among these constraints, the relations $\accentset{s}{\mathcal{Z}}^{(11)}_{13}= \accentset{s}{\mathcal{Z}}^{(11)}_{31}$ and $\accentset{s}{\mathcal{Z}}^{(11)}_{23}= \accentset{s}{\mathcal{Z}}^{(11)}_{32}$ reduce to the linear system
%-----------------------------
\begin{equation} 
\begin{dcases}
	C_1f_4(R)+C_2f_6(R)=0\,,\\
	C_3f_4(R)+C_4f_6(R)=0\,.
\end{dcases}
\end{equation}
%-----------------------------
Since $K=C_1C_4-C_2C_3\neq0$, the coefficient matrix is nonsingular, and hence the unique solution is
%-----------------------------
\begin{equation}
	f_4(R)=f_6(R)=0\,.
\end{equation}
%-----------------------------
No additional conditions arise from the remaining universality PDEs, which are satisfied identically. Thus $f_1(R)$, $f_2(R)$, $f_3(R)$, and $f_5(R)$ remain arbitrary, and the universal residual stress fields take the form
%-----------------------------
\begin{equation} 
	[\mathring{S}^{AB}]
	=
	\begin{bmatrix}
	A(R) & 0 & 0\\
	0 & B(R) & D(R) \\
	0 & D(R) & C(R)
	\end{bmatrix}
	\,,
\end{equation}
%-----------------------------
where $A$, $B$, $C$, and $D$ are arbitrary twice-differentiable functions of $R$.

%-----------------------------
\begin{example}
A subset of Family $3$ was studied by \citet{MerodioOgden2016} in their analysis of the inflation, extension, and torsion of a residually stressed incompressible circular cylindrical tube. Their deformation is given by
%-----------------------------
\begin{equation}
	r=\sqrt{a^2+\lambda_z^{-1}(R^2-A^2)}\,,\qquad
	\theta=\Theta+\psi\lambda_z Z\,,\qquad
	z=\lambda_z Z\,.
\end{equation}
%-----------------------------
Comparing this deformation with \eqref{eqn:Family3-u}, one concludes that $C_1=1$, $C_2=\psi\lambda_z$, $C_3=0$, $C_4=\lambda_z$, and $C_6=C_7=0$, and consequently
%-----------------------------
\begin{equation}
	K=C_1C_4-C_2C_3=\lambda_z\,,\qquad
	C_5=a^2-\frac{A^2}{\lambda_z}\,.
\end{equation}
%-----------------------------
More specifically, the deformation considered by \citet{MerodioOgden2016} is a three-parameter subfamily of Family $3$. In the notation used here, it is obtained by setting $C_1=1$ and $C_3=0$, leaving $C_2$, $C_4$, and $C_5$ as the independent parameters.

The residual stress field considered by \citet{MerodioOgden2016} is more restrictive than the one considered in the present work. They assume from the outset that the residual stress tensor is diagonal in cylindrical coordinates and depends only on the radial coordinate $R$:
%-----------------------------
\begin{equation}
	[\mathring{S}^{AB}]
	=
	\begin{bmatrix}
		\mathring{S}^{RR}(R) & 0 & 0\\
		0 & \mathring{S}^{\Theta\Theta}(R) & 0\\
		0 & 0 & \mathring{S}^{ZZ}(R)
	\end{bmatrix}\,.
\end{equation}
%-----------------------------
The traction-free boundary conditions on the ends of the tube imply that $\mathring{S}^{ZZ}(R)=0$. Note that the traction-free boundary conditions on the cylindrical boundaries involve only the components $\mathring{S}^{RR}$, $\mathring{S}^{R\Theta}$, and $\mathring{S}^{RZ}$. However, we have shown that for any universal residual stress field $\mathring{S}^{R\Theta}(R)=\mathring{S}^{RZ}(R)=0$. 
Nevertheless, the residual stress field considered by \citet{MerodioOgden2016} is not the most general residual stress field compatible with the symmetry of the deformation and the traction-free boundary conditions on the ends of the tube. In particular, a non-vanishing $\Theta Z$ shear stress component is also admissible. The most general residual stress field compatible with the symmetry of the deformation and the traction-free boundary conditions on the ends of the tube is
%-----------------------------
\begin{equation}
	[\mathring{S}^{AB}]
	=
	\begin{bmatrix}
		\mathring{S}^{RR}(R) & 0 & 0\\
		0 & \mathring{S}^{\Theta\Theta}(R) & \mathring{S}^{\Theta Z}(R)\\
		0 & \mathring{S}^{\Theta Z}(R) & 0
	\end{bmatrix}\,.
\end{equation}
%-----------------------------
Denoting the inner and outer radii of the tube by $R_i$ and $R_o$, respectively, the traction boundary conditions imply that $\mathring{S}^{RR}(R_i)=\mathring{S}^{RR}(R_o)=0$.
\end{example}
%-----------------------------

%-----------------------------
%-----------------------------
\subsection{Family 4: Inflation/inversion of a sector of a spherical shell}

With respect to the spherical coordinates $(R,\Theta,\Phi)$ and $(r,\theta,\phi)$ in the reference and current configurations, respectively, this family of deformations have the representation \eqref{eqn:Family4-u}.
The rotational symmetry of $\mathbf{C}^\flat$ about the radial direction $\hat{\mathbf{R}}$ suggests the following ansatz for the residual stress field:
%-----------------------------
\begin{equation} 
	\mathring{\mathbf{S}}(\mathbf{X})=f_1(R)\,\hat{\mathbf{R}}\otimes\hat{\mathbf{R}}
	+f_2(R)\,(\mathbf{1}-\hat{\mathbf{R}}\otimes\hat{\mathbf{R}})	\,,
\end{equation}
%-----------------------------
where $f_1(R)$ and $f_2(R)$ are arbitrary twice-differentiable functions. In spherical coordinates, this ansatz has the representation
%-----------------------------
\begin{equation} \label{Residual-Strss-Family4}
	[\mathring{S}^{AB}]
	=\begin{bmatrix}
	f_1(R) & 0 & 0\\
	0 & \frac{f_2(R)}{R^2} & 0\\
	0 & 0 & \frac{f_2(R)}{R^2\sin^2\Theta}
	\end{bmatrix}
	\,.
\end{equation}
%-----------------------------
Substitution into the universality PDEs shows that all constraints are satisfied identically, and hence no further restrictions are imposed on $f_1(R)$ and $f_2(R)$. Therefore, \eqref{Residual-Strss-Family4} characterizes the most general universal residual stress field for this family.

%-----------------------------
%-----------------------------
\subsection{Family 5: Inflation, bending, extension, and azimuthal shearing of an annular wedge}

With respect to the cylindrical coordinates $(R,\Theta,Z)$ and $(r,\theta,z)$ in the reference and current configurations, respectively, this family of deformations have the representation \eqref{eqn:Family5-u}.
Motivated by the radial symmetry of the deformation, we consider the following general symmetric residual stress field:
%-----------------------------
\begin{equation} 
	[\mathring{S}^{AB}]
	=
	\begin{bmatrix}
	f_1(R) & f_4(R) & f_6(R)\\
	f_4(R) & f_2(R) & f_5(R) \\
	f_6(R) & f_5(R) & f_3(R)
	\end{bmatrix}
	\,,
\end{equation}
%-----------------------------
where $f_i(R)$, $i=1,\ldots,6$, are unknown functions to be determined from the universality constraints.

The universality constraint $\mathcal{Z}_{13}=\mathcal{Z}_{31}$ reduces to the ordinary differential equation
%-----------------------------
\begin{equation}
	R^2 f_6''(R)+R f_6'(R)-f_6(R)=0\,.
\end{equation}
%-----------------------------
This Cauchy--Euler equation has the general solution
%-----------------------------
\begin{equation}
	f_6(R)=A_1 R+\frac{A_2}{R}\,,
\end{equation}
%-----------------------------
where $A_1$ and $A_2$ are arbitrary constants.
Similarly, the universality constraint $\mathcal{X}_{13}=\mathcal{X}_{31}$ reduces to
%-----------------------------
\begin{equation}
	3f_5'(R)+Rf_5''(R)=0\,,
\end{equation}
%-----------------------------
whose general solution is
%-----------------------------
\begin{equation}
	f_5(R)=K_1+\frac{K_2}{R^2}\,,
\end{equation}
%-----------------------------
where $K_1$ and $K_2$ are arbitrary constants.

The universality constraint $\accentset{s}{\mathcal{C}}^{(44)}_{12}= \accentset{s}{\mathcal{C}}^{(44)}_{21}$ yields the ordinary differential equation
%-----------------------------
\begin{equation}
    C_1^2 f_1'(R)+\frac{f_3'(R)}{C_1^4 C_3^2}
    +C_1^2\Big(2 C_3^2 R f_2(R)+2 C_2 C_3 f_4(R)+C_2^2 f_1'(R)+C_3^2 R^2 f_2'(R)+2 C_2 C_3 R f_4'(R)\Big)
    =0\,,
\end{equation}
%-----------------------------
which can be rearranged as
%-----------------------------
\begin{equation} \label{ODE-Family5}
    R^2 f_2'(R)+2 R f_2(R)=  -\frac{1+C_2^2}{C_3^2}f_1'(R) -\frac{f_3'(R)}{C_1^6 C_3^4}
    -\frac{2 C_2}{C_3}f_4(R) -\frac{2 C_2}{C_3}R f_4'(R)    \,.
\end{equation}
%-----------------------------
Noting that
%-----------------------------
\begin{equation}
    \frac{d}{dR}\!\left(R^2 f_2(R)\right) =  R^2 f_2'(R) +2R f_2(R) \,,\qquad
    \frac{d}{dR}\!\left(R f_4(R)\right) = f_4(R) +R f_4'(R)   \,,
\end{equation}
%-----------------------------
equation \eqref{ODE-Family5} can be written as
%-----------------------------
\begin{equation}
    \frac{d}{dR}\!\left(R^2 f_2(R)\right)=  -\frac{1+C_2^2}{C_3^2}f_1'(R)
    -\frac{f_3'(R)}{C_1^6 C_3^4} -\frac{2 C_2}{C_3} \frac{d}{dR}\!\left(R f_4(R)\right)
    \,.
\end{equation}
%-----------------------------
Integrating with respect to $R$ yields
%-----------------------------
\begin{equation}
    f_2(R)=\frac{K_3}{R^2} -\frac{(1+C_2^2)f_1(R)}{C_3^2 R^2}
    -\frac{f_3(R)}{C_1^6 C_3^4 R^2} -\frac{2 C_2 f_4(R)}{C_3 R}    \,,
\end{equation}
%-----------------------------
where $K_3$ is an integration constant.

The universality constraint $\mathcal{Z}_{12}=\mathcal{Z}_{21}$ reduces to
%-----------------------------
\begin{equation}
	-C_1\Big[3 C_3 f_4(R)+3 C_2 f_1'(R)+R\big(5 C_3 f_4'(R)+C_2 f_1''(R)+C_3 R f_4''(R)\big)\Big]=0	\,.
\end{equation}
%-----------------------------
Since $C_1\neq0$, this equation is equivalent to
%-----------------------------
\begin{equation}
	C_3\big[R^2 f_4''(R)+5R f_4'(R)+3f_4(R)\big]+C_2\big[R f_1''(R)+3f_1'(R)\big]=0	\,.
\end{equation}
%-----------------------------
Solving for $f_4(R)$, one obtains
%-----------------------------
\begin{equation}
	f_4(R)=-\frac{C_2 f_1(R)}{C_3R}
	+\frac{B_1}{2C_3R}
	+\frac{B_2}{R^3}
	\,,
\end{equation}
%-----------------------------
where $B_1$ and $B_2$ are arbitrary constants.

The universality constraint $\mathcal{W}_{12}=\mathcal{W}_{21}$ reduces to
%-----------------------------
\begin{equation} \label{W12-W21}
	(2B_2C_3R+B_1R^3)f_3''(R)+(-2B_2C_3+3B_1R^2)f_3'(R)
	-16A_1(A_1C_2+C_3K_1)R^3=0	\,.
\end{equation}
%-----------------------------
For $B_1\neq0$, the general solution is
%-----------------------------
\begin{equation}
	f_3(R)=
	D_1+\frac{2A_1(A_1C_2+C_3K_1)}{B_1}R^2
	+\frac{D_2}{B_1R^2+2B_2C_3}
	\,,
\end{equation}
%-----------------------------
where $D_1$ and $D_2$ are arbitrary constants.
The universality constraint $\mathcal{X}_{12}=\mathcal{X}_{21}$ further implies that
%-----------------------------
\begin{equation}
	(A_1^2C_2+A_1C_3K_1)(2B_2C_3+B_1R^2)^3
	-B_1^2B_2C_3D_2=0	\,.
\end{equation}
%-----------------------------
We first consider the case $B_2\neq0$. Since $B_1\neq0$ and the above equation must hold for all $R$, the $R$-dependent term must vanish. Thus,
%-----------------------------
\begin{equation}
	A_1(A_1C_2+C_3K_1)=0\,.
\end{equation}
%-----------------------------
Hence, either $A_1=0$ or $K_1=-\frac{A_1C_2}{C_3}$. The remaining constant term implies that
%-----------------------------
\begin{equation}
	-B_1^2B_2C_3D_2=0\,.
\end{equation}
%-----------------------------
Since $B_1\neq0$, $B_2\neq0$, and $C_3\neq0$, it follows that $D_2=0$. Therefore, in this case there are two possibilities:
Case (i) $A_1=0$ and $D_2=0$;
Case (ii) $K_1=-\frac{A_1C_2}{C_3}$ and $D_2=0$.

Next, consider the case $B_2=0$. The universality constraint reduces to
%-----------------------------
\begin{equation}
	(A_1^2C_2+A_1C_3K_1)(B_1R^2)^3=0\,.
\end{equation}
%-----------------------------
Since $B_1\neq0$, this equation is equivalent to
%-----------------------------
\begin{equation}
	A_1(A_1C_2+C_3K_1)=0\,.
\end{equation}
%-----------------------------
Hence, either $A_1=0$ or $K_1=-\frac{A_1C_2}{C_3}$. In contrast to the previous case, no restriction is imposed on $D_2$, and therefore $D_2$ remains arbitrary. Thus, there are two possibilities in this case:
Case (iii) $B_2=0$, $A_1=0$;
Case (iv) $B_2=0$, $K_1=-\frac{C_2}{C_3}A_1$.

%-----------------------------
\begin{itemize}[topsep=0pt,noitemsep, leftmargin=10pt]
\item Case (i) $A_1=0$, $D_2=0$: The universality constraint $\mathcal{W}_{13}=\mathcal{W}_{31}$ reduces to
%-----------------------------
\begin{equation}
	8B_2C_3(A_2C_2+C_3K_2)-A_2R^3f_1'(R)+A_2R^4f_1''(R)=0	\,.
\end{equation}
%-----------------------------
If $A_2=0$, this equation reduces to $8B_2C_3^2K_2=0$. Since $B_2\neq0$ and $C_3\neq0$, it follows that $K_2=0$.
The universality constraint $\mathcal{K}_{13}=\mathcal{K}_{31}$ reduces to
%-----------------------------
\begin{equation}
	K_1\Big(3f_1'(R)+Rf_1''(R)\Big)=0\,.
\end{equation}
%-----------------------------
If $K_1=0$, the equation is identically satisfied. The universality constraint $\mathcal{B}^{(7)}_{12} =\mathcal{B}^{(7)}_{21}$ then implies that $3f_1'(R)+Rf_1''(R)=0$, and hence
%-----------------------------
\begin{equation}\label{f1-solution}
	f_1(R)=E_1+\frac{E_2}{R^2}	\,,
\end{equation}
%-----------------------------
where $E_1$ and $E_2$ are arbitrary constants. If $K_1\neq0$, then $3f_1'(R)+Rf_1''(R)=0$, which again has the solution \eqref{f1-solution}.

The universality constraint $\mathcal{B}^{(5)}_{12} =\mathcal{B}^{(5)}_{21}$ now reduces to $B_2^2C_3^2+E_2^2=0$, and hence $B_2=E_2=0$. Therefore, when $A_2=0$, Case (i) is not possible.

\vskip 0.1in
If $A_2\neq0$, dividing by $A_2$ gives
%-----------------------------
\begin{equation}
	R^4f_1''(R)-R^3f_1'(R)
	=-\frac{8B_2C_3(A_2C_2+C_3K_2)}{A_2}\,.
\end{equation}
%-----------------------------
Integrating twice, one obtains
%-----------------------------
\begin{equation}
	f_1(R)= E_1+E_2R^2+\frac{B_2C_3(A_2C_2+C_3K_2)}{A_2R^2}	\,,
\end{equation}
%-----------------------------
where $E_1$ and $E_2$ are arbitrary constants. Next, the universality constraint $\mathcal{B}^{(7)}_{12} =\mathcal{B}^{(7)}_{21}$ implies that $E_2=0$. The universality constraint $\mathcal{W}_{13}=\mathcal{W}_{31}$ further reduces to $B_2(A_2C_2+K_2C_3)=0$. Since $C_2\neq0$, one obtains $A_2=-\frac{C_3}{C_2}K_2$.

The universality constraint $\mathcal{X}^{(7)}_{12} =\mathcal{X}^{(7)}_{21}$ then implies that $B_2^2 C_1^7 C_2 C_3^2 \left(1+C_2^2+C_3^2\right)=0$, and hence $B_2=0$. Therefore, Case (i) is not possible, i.e., one must have $B_2=0$.

\item Case (ii) $K_1=-\frac{C_2}{C_3}A_1$, $D_2=0$: 
The universality constraint $\mathcal{B}^{(7)}_{12} =\mathcal{B}^{(7)}_{21}$ implies that $3f_1'(R)+Rf_1''(R)=0$, and hence $f_1(R)$ is given by \eqref{f1-solution}. The universality constraint $\mathcal{B}^{(4)}_{12} =\mathcal{B}^{(4)}_{21}$ reduces to
%-----------------------------
\begin{equation}
	C_2\left(B_2^2 C_1^6 C_3^4+C_1^6 C_3^2 E_2^2+A_1^2 R^6\right)=0\,,
\end{equation}
%-----------------------------
and therefore $B_2=E_2=A_1=0$. Thus, Cases (i) and (ii) are not possible, and one must have $B_2=0$.

\item Case (iii) $B_2=0$, $A_1=0$: The universality constraint $\mathcal{B}^{(7)}_{12} =\mathcal{B}^{(7)}_{21}$ implies that $3f_1'(R)+Rf_1''(R)=0$, and hence $f_1(R)$ is given by \eqref{f1-solution}. The universality constraint $\mathcal{W}_{13}=\mathcal{W}_{31}$ reduces to $A_2\left(D_2+B_1C_1^6C_3^2E_2\right)=0$, and therefore either $A_2=0$ or $D_2=-B_1C_1^6C_3^2E_2$.

If $A_2=0$, the universality constraint $\mathcal{K}_{13}=\mathcal{K}_{31}$ implies that either $E_2=0$ or $K_2=0$.
If $E_2=0$, the universality constraints $\mathcal{B}^{(4)}_{12} =\mathcal{B}^{(4)}_{21}$ and $\mathcal{B}^{(4)}_{23} =\mathcal{B}^{(4)}_{32}$ imply that $D_2=K_1=0$. Then, $\mathcal{B}^{(55)}_{12} =\mathcal{B}^{(55)}_{21}$ implies that $K_2=0$. The remaining universality constraints are identically satisfied.
If $K_2=0$, the universality constraint $\mathcal{K}_{12}=\mathcal{K}_{21}$ implies that either $B_1=0$ or $K_1=0$. Since we are assuming $B_1\neq0$, it follows that $K_1=0$. The universality constraint $\mathcal{K}_{13}=\mathcal{K}_{31}$ then implies that $E_2=0$. Next, $\mathcal{B}^{(4)}_{12} =\mathcal{B}^{(4)}_{21}$ implies that $D_2=0$. The remaining universality constraints are identically satisfied.

If $D_2=-B_1C_1^6C_3^2E_2$, the universality constraint $\mathcal{B}^{(5)}_{12} =\mathcal{B}^{(5)}_{21}$ implies that $E_2=K_1=0$. Next, $\mathcal{B}^{(55)}_{12} =\mathcal{B}^{(55)}_{21}$ implies that $A_2=K_2=0$, which reduces to the previous case.

In summary, the universal residual stress field in this case has the following representation:
%-----------------------------
\begin{equation} \label{Universal-ResidualStress-1}
	[\mathring{S}^{AB}]
	= \begin{bmatrix}
		E_1 &\frac{B_1-2C_2E_1}{2C_3R} & 0\\
		\frac{B_1-2C_2E_1}{2C_3R} &
		-\frac{B_1C_1^6C_2C_3^2+D_1
			+C_1^6C_3^2(E_1-C_2^2E_1-C_3^2K_3)}{C_1^6C_3^4R^2} & 0\\
		0 & 0 & D_1
	\end{bmatrix}
	\,.
\end{equation}
%-----------------------------

\item Case (iv) $B_2=0$, $K_1=-\frac{C_2}{C_3}A_1$: 
The universality constraint $\mathcal{B}^{(7)}_{12} =\mathcal{B}^{(7)}_{21}$ implies that $3f_1'(R)+Rf_1''(R)=0$, and hence $f_1(R)$ is given by \eqref{f1-solution}. The universality constraint $\mathcal{Z}^{(5)}_{12} =\mathcal{Z}^{(5)}_{21}$ implies that $A_1=0$. Next, $\mathcal{W}_{13}=\mathcal{W}_{31}$ implies that $A_2\left(D_2+B_1C_1^6C_3^2E_2\right)=0$, and therefore either $A_2=0$ or $D_2=-B_1C_1^6C_3^2E_2$.
If $A_2=0$, $\mathcal{C}^{(5)}_{12} =\mathcal{C}^{(5)}_{21}$ implies that $D_2=0$. Similarly, $\mathcal{B}^{(5)}_{12} =\mathcal{B}^{(5)}_{21}$ implies that $E_2=0$. In addition, $\mathcal{C}^{(59)}_{12} =\mathcal{C}^{(59)}_{21}$ implies that $K_2=0$. The remaining universality constraints are identically satisfied.
If $D_2=-B_1C_1^6C_3^2E_2$, the universality constraint $\mathcal{B}^{(5)}_{12} =\mathcal{B}^{(5)}_{21}$ implies that $E_2=0$. Then, $\mathcal{Z}^{(5)}_{13} =\mathcal{Z}^{(5)}_{31}$ implies that either $D_2=0$ or $A_2=0$.
If $A_2=0$, $\mathcal{B}^{(55)}_{12} =\mathcal{B}^{(55)}_{21}$ implies that $K_2=0$. The remaining universality constraints are identically satisfied.
If $D_2=0$, $\mathcal{B}^{(55)}_{12} =\mathcal{B}^{(55)}_{21}$ implies that $A_2=0$, which reduces to the previous case.

\end{itemize}
%-----------------------------

\vskip 0.1in \noindent
Now let us return to \eqref{W12-W21} and assume that $B_1=0$. Thus,
%-----------------------------
\begin{equation} \label{W12-W21-reduced}
	2B_2C_3R f_3''(R)-2B_2C_3 f_3'(R)
	-16A_1(A_1C_2+C_3K_1)R^3=0
	\,.
\end{equation}
%-----------------------------
We first consider the case $B_2\neq0$. Since $C_3\neq0$, equation \eqref{W12-W21-reduced} is equivalent to
%-----------------------------
\begin{equation}
	Rf_3''(R)-f_3'(R)
	=\frac{8A_1(A_1C_2+C_3K_1)}{B_2C_3}R^3
	\,.
\end{equation}
%-----------------------------
The general solution is
%-----------------------------
\begin{equation}
	f_3(R)=D_1+D_2R^2
	+\frac{A_1(A_1C_2+C_3K_1)}{B_2C_3}R^4
	\,,
\end{equation}
%-----------------------------
where $D_1$ and $D_2$ are arbitrary constants.
The universality constraint $\mathcal{X}_{12}=\mathcal{X}_{21}$ implies that $D_2=0$. The universality constraint $\mathcal{H}_{12}=\mathcal{H}_{21}$ implies that either $A_1=0$ or $K_1=-\frac{A_1C_2}{C_3}$.

If $A_1=0$, the universality constraint $\mathcal{C}^{(7)}_{12} =\mathcal{C}^{(7)}_{21}$ implies that $3f_1'(R)+Rf_1''(R)=0$, and hence $f_1(R)$ is given by \eqref{f1-solution}. Next, $\mathcal{B}^{(9)}_{12} =\mathcal{B}^{(9)}_{21}$ implies that $B_2=0$, which is a contradiction.

If $K_1=-\frac{A_1C_2}{C_3}$, the universality constraint $\mathcal{C}^{(8)}_{12} =\mathcal{C}^{(7)}_{21}$ implies that $3f_1'(R)+Rf_1''(R)=0$, and hence $f_1(R)$ is given by \eqref{f1-solution}. From $\mathcal{K}_{12} =\mathcal{K}_{21}$ one concludes that $A_1=0$. Next, the universality constraint $\mathcal{B}^{(5)}_{12} =\mathcal{B}^{(5)}_{21}$ implies that $B_2=E_2=0$, which is a contradiction.

Now knowing that $B_2=0$, equation \eqref{W12-W21-reduced} reduces to $A_1(A_1C_2+C_3K_1)=0$. Thus, either $A_1=0$ or $K_1=-\frac{A_1C_2}{C_3}$.
If $A_1=0$, the universality constraint $\mathcal{X}_{12} =\mathcal{X}_{21}$ implies that $3f_3'(R)+Rf_3''(R)=0$, and hence
%-----------------------------
\begin{equation} \label{f3-solution}
	f_3(R)=D_1+\frac{D_2}{R^2}	\,,
\end{equation}
%-----------------------------
where $D_1$ and $D_2$ are arbitrary constants. Next, the universality constraint $\mathcal{B}^{(7)}_{12} =\mathcal{B}^{(7)}_{21}$ implies that $3f_1'(R)+Rf_1''(R)=0$, and hence $f_1(R)$ is given by \eqref{f1-solution}. From $\mathcal{C}^{(5)}_{12} =\mathcal{C}^{(5)}_{21}$ one concludes that $K_1=0$.
The universality constraint $\mathcal{B}^{(5)}_{12} =\mathcal{B}^{(5)}_{21}$ implies that $D_2=E_2=0$. Next, $\mathcal{Z}^{(5)}_{13} =\mathcal{Z}^{(5)}_{31}$ implies that $A_2=0$. Also, $\mathcal{B}^{(55)}_{12} =\mathcal{B}^{(55)}_{21}$ implies that $K_2=0$. The remaining universality constraints are identically satisfied. Therefore, the universal residual stresses have the form
%-----------------------------
\begin{equation} \label{Universal-ResidualStress-2}
	[\mathring{S}^{AB}]	=
	\begin{bmatrix}
		E_1 & -\dfrac{C_2E_1}{C_3R} & 0\\
		-\dfrac{C_2E_1}{C_3R} &
		\dfrac{-D_1+C_1^6C_3^2\Big[(-1+C_2^2)E_1+C_3^2K_3\Big]}{C_1^6C_3^4R^2} & 0\\
		0 & 0 & D_1
	\end{bmatrix}
	\,.
\end{equation}
%-----------------------------

If $K_1=-\frac{A_1C_2}{C_3}$, the universality constraint $\mathcal{X}_{12} =\mathcal{X}_{21}$ implies that $3f_3'(R)+Rf_3''(R)=0$, and hence $f_3(R)$ is given by \eqref{f3-solution}. Next, the universality constraint $\mathcal{B}^{(7)}_{12} =\mathcal{B}^{(7)}_{21}$ implies that $3f_1'(R)+Rf_1''(R)=0$, and hence $f_1(R)$ is given by \eqref{f1-solution}.
From $\mathcal{C}^{(5)}_{12} =\mathcal{C}^{(5)}_{21}$ one concludes that $A_1=D_2=E_2=0$. Next, $\mathcal{B}^{(55)}_{12} =\mathcal{B}^{(55)}_{21}$ implies that $A_2=K_2=0$. The remaining universality constraints are identically satisfied. Thus, the universal residual stress fields are identical to those in the previous case, namely \eqref{Universal-ResidualStress-2}.
Note that the residual stress field \eqref{Universal-ResidualStress-2} is obtained from \eqref{Universal-ResidualStress-1} by setting $B_1=0$. Therefore, \eqref{Universal-ResidualStress-2} is a special case of \eqref{Universal-ResidualStress-1}.

Table \ref{table:UniversalResidualStresses} summarizes the main results of this section. For each of the six known families of universal deformations, it lists the corresponding right Cauchy--Green tensor and the associated class of universal residual stress fields.

\vskip 0.2in
%-----------------------------
\begin{table}[hbt!]
\begin{center}
\resizebox{\textwidth}{!}{
\begin{tabular}{|c|c|c|c|}
\hline
Family & Universal Deformations & $\mathbf{C}^{\flat}$ & Universal residual stresses \\
\hline
$0$ & Homogeneous deformations & $[C_{AB}]=\text{constant}$ & $[\mathring{S}^{AB}]=\text{constant}$ \\
\hline
$1$ &
$\begin{cases}
	r(X,Y,Z)=\sqrt{C_1(2X+C_4)}\\
	\theta(X,Y,Z)=C_2(Y+C_5)\\
	z(X,Y,Z)=\frac{Z}{C_1C_2}-C_2C_3Y+C_6
\end{cases}$
&
$[C_{AB}]=\begin{bmatrix}
	\frac{C_1}{2X+C_4} & 0 & 0\\
	0 & C_2^2\left[C_1(2X+C_4)+C_3^2\right] & -\frac{C_3}{C_1}\\
	0 & -\frac{C_3}{C_1} & \frac{1}{C_1^2C_2^2}
\end{bmatrix}$
&
$[\mathring{S}^{AB}]=
\begin{bmatrix}
	A(X) & 0 & 0\\
	0 & B(X) & D(X)\\
	0 & D(X) & C(X)
\end{bmatrix}$
\\
\hline
$2$ &
$\begin{cases}
	x(R,\Theta,Z)=\frac{1}{2}C_1C_2^2R^2+C_4\\
	y(R,\Theta,Z)=\frac{\Theta}{C_1C_2}+C_5\\
	z(R,\Theta,Z)=\frac{C_3}{C_1C_2}\Theta+\frac{1}{C_2}Z+C_6
\end{cases}$
&
$[C_{AB}]=\begin{bmatrix}
	C_1^2C_2^4R^2 & 0 & 0\\
	0 & \frac{C_3^2+1}{C_1^2C_2^2} & \frac{C_3}{C_1C_2^2}\\
	0 & \frac{C_3}{C_1C_2^2} & \frac{1}{C_2^2}
\end{bmatrix}$
&
$[\mathring{S}^{AB}]=
\begin{bmatrix}
	A(R) & 0 & 0\\
	0 & B(R) & D(R)\\
	0 & D(R) & C(R)
\end{bmatrix}$
\\
\hline
$3$ &
$\begin{cases}
	r(R,\Theta,Z)=\sqrt{\frac{R^2}{K}+C_5}\\
	\theta(R,\Theta,Z)=C_1\Theta+C_2Z+C_6\\
	z(R,\Theta,Z)=C_3\Theta+C_4Z+C_7
\end{cases}$
&
$[C_{AB}]=\begin{bmatrix}
	\frac{R^2}{K(KC_5+R^2)} & 0 & 0\\
	0 & C_1^2\left(\frac{R^2}{K}+C_5\right)+C_3^2
	& C_1C_2\left(\frac{R^2}{K}+C_5\right)+C_3C_4\\
	0 & C_1C_2\left(\frac{R^2}{K}+C_5\right)+C_3C_4
	& C_2^2\left(\frac{R^2}{K}+C_5\right)+C_4^2
\end{bmatrix}$
&
$[\mathring{S}^{AB}]=
\begin{bmatrix}
	A(R) & 0 & 0\\
	0 & B(R) & D(R)\\
	0 & D(R) & C(R)
\end{bmatrix}$
\\
\hline
$4$ &
$\begin{cases}
	r(R,\Theta,\Phi)=(\pm R^3+C_1^3)^{\frac{1}{3}}\\
	\theta(R,\Theta,\Phi)=\pm\Theta\\
	\phi(R,\Theta,\Phi)=\Phi
\end{cases}$
&
$[C_{AB}]=\begin{bmatrix}
	\frac{R^4}{\left(C_1^3\pm R^3\right)^{4/3}} & 0 & 0\\
	0 & \left(C_1^3\pm R^3\right)^{2/3} & 0\\
	0 & 0 & \left(C_1^3\pm R^3\right)^{2/3}\sin^2\Theta
\end{bmatrix}$
&
$[\mathring{S}^{AB}]=
\begin{bmatrix}
	f_1(R) & 0 & 0\\
	0 & \frac{f_2(R)}{R^2} & 0\\
	0 & 0 & \frac{f_2(R)}{R^2\sin^2\Theta}
\end{bmatrix}$
\\
\hline
$5$ &
$\begin{cases}
	r(R,\Theta,Z)=C_1R\\
	\theta(R,\Theta,Z)=C_2\log R+C_3\Theta+C_4\\
	z(R,\Theta,Z)=\frac{1}{C_1^2C_3}Z+C_5
\end{cases}$
&
$[C_{AB}]=\begin{bmatrix}
	C_1^2(C_2^2+1) & C_1^2C_2C_3R & 0\\
	C_1^2C_2C_3R & C_1^2C_3^2R^2 & 0\\
	0 & 0 & \frac{1}{C_1^4C_3^2}
\end{bmatrix}$
&
$[\mathring{S}^{AB}]=
\begin{bmatrix}
	A_1 & \frac{A_2}{2C_3R} & 0\\
	\frac{A_2}{2C_3R} & \frac{A}{R^2} & 0\\
	0 & 0 & A_4
\end{bmatrix}$
\\
\hline
\end{tabular}}
\end{center}
\vskip -0.1in
\caption[]{Universal residual stress fields corresponding to the six known families of universal deformations of incompressible isotropic elasticity. The presence of residual stress does not enlarge the class of universal deformations. For Family $3$, $K=C_1C_4-C_2C_3\neq 0$.
Also, for Family $5$, $A=-\frac{A_2C_1^6C_2C_3^2	+A_4+C_1^6C_3^2(1+C_2^2)A_1
-C_1^6C_3^4A_3}{C_1^6C_3^4}$.}
\label{table:UniversalResidualStresses}
\end{table}
%-----------------------------

%-----------------------------
\begin{remark}
It is worth noting that, for Families $1$--$4$, the universal residual stress fields obtained in the present formulation have the same structural form as the corresponding universal material metrics of \citet{Goodbrake2020}. This correspondence does not hold for Family $5$.
At first sight, this discrepancy may appear surprising because universality should be an intrinsic property of a deformation and therefore independent of the particular representation used to model material inhomogeneity. However, the two theories are based on different constitutive classes. In \citep{Goodbrake2020}, the material metric is the primary unknown and residual stress is induced indirectly through the constitutive law, whereas in the present theory the residual stress is treated as an independent constitutive field. Consequently, the universality constraints obtained here are substantially stronger and lead to a larger system of universality PDEs. The fact that the two approaches nevertheless produce identical universal deformations is analogous to the well-known situation in classical elasticity where Cauchy elasticity and hyperelasticity possess the same universal deformations despite representing different constitutive classes \citep{Yavari2024Cauchy,MotaghianYavari2026}. The discrepancy observed for the residual stress fields of Family $5$ therefore does not indicate an inconsistency between the two theories. Rather, it suggests that the residual-stress formulation imposes additional restrictions on the admissible residual stress fields associated with a given universal deformation. In particular, while a universal deformation of Family $5$ may be supported by a broad family of material metrics, only a smaller subclass of the corresponding induced residual stress fields remains universal when residual stress is treated as an independent constitutive variable.
\end{remark}
%-----------------------------

%-----------------------------
%-----------------------------
\section{Conclusions}  \label{Sec:Conclusions}

We studied universal deformations in incompressible isotropic Cauchy elastic solids with residual stress. Starting from the constitutive representation of the Cauchy stress as an isotropic tensor-valued function of the deformation and residual stress, we derived the universality constraints for residually-stressed incompressible isotropic Cauchy elastic solids. These constraints generalize the classical universality equations of incompressible isotropic elasticity by incorporating the residual stress field as an additional constitutive variable.

A central result of this work is that residual stress does not enlarge the class of universal deformations. We show that, for each of the six known families of universal deformations, the set of universal deformations of incompressible isotropic Cauchy elastic solids with residual stress coincides with the corresponding set in the absence of residual stress. Thus, the presence of residual stress does not alter the class of universal deformations.

Motivated by this result, we determined the residual stress fields that are compatible with the six known families of universal deformations. Assuming that the residual stress field has the same symmetry as the corresponding universal deformation, the universality constraints reduce to systems of ordinary differential equations. These equations were solved explicitly for each family, and the corresponding universal residual stress fields were determined. 
The resulting universal residual stress fields reflect the symmetry structure of the associated universal deformation. For Families $1$, $2$, and $3$, they are characterized by four arbitrary functions of one variable; for Family $4$, by two arbitrary functions of one variable; and for Families $0$ and $5$, by a finite number of constants.

The present work extends the classical theory of universal deformations to residually-stressed incompressible isotropic Cauchy elastic solids. It also complements earlier studies of universal eigenstrains and universal inhomogeneities by addressing residual stress directly, without introducing an underlying material metric or eigenstrain distribution. The universality constraints derived here provide a framework for investigating more general classes of residually-stressed solids, including anisotropic materials and solids with additional internal constraints.

\bibliographystyle{abbrvnat}
\bibliography{ref,ref1}

\end{document}